\documentclass[10pt]{iopart}

\makeatletter
\expandafter\let\csname equation*\endcsname\relax
\expandafter\let\csname endequation*\endcsname\relax
\makeatother

\usepackage{amsmath}
\usepackage{amssymb}
\usepackage{amsthm}

\usepackage{bm}
\usepackage{booktabs}
\usepackage{comment}

\usepackage{enumerate}

\usepackage{graphicx}
\graphicspath{{./}{Figs/}}
\usepackage[colorlinks=True,linkcolor=blue,citecolor=blue,urlcolor=blue]{hyperref}
\usepackage{mathtools}

\usepackage{subcaption}
\usepackage{xcolor}
\usepackage{tabularx}
\usepackage{mathrsfs}
\mathtoolsset{showonlyrefs}

\usepackage{tikz}
\usetikzlibrary{positioning,arrows.meta}

\DeclareMathOperator{\diag}{diag}
\DeclareMathOperator{\spn}{span}
\DeclareMathOperator{\spec}{spec}

\definecolor{nullred}{HTML}{9F2305}

\newcommand{\abs}[1]{\left\vert #1\right\vert}
\newcommand{\comm}[2]{\left[{#1} , {#2}\right]}
\newcommand{\acomm}[2]{\left\{{#1} , {#2}\right\}}
\newcommand{\dv}[2]{\frac{\mathrm{d}{#1}}{\mathrm{d}{#2}}}
\newcommand{\pdv}[2]{\frac{\partial {#1}}{\partial{#2}}}
\newcommand{\ip}[2]{\langle {#1} \vert {#2} \rangle}
\newcommand{\ev}[1]{\langle {#1}  \rangle}
\newcommand{\ket}[1]{\left\vert {#1}\right>}

\newcommand{\vb}[1]{\mathbf{#1}}
\newcommand{\dd}{\mathrm{d}}
\newcommand{\cl}{\mathrm{Cl}}

\newcommand{\ca}[1]{\mathrm{Cl}_{#1}(\mathbb{R})}

\newcommand{\iu}{\mathrm{i}}

\newcommand{\bl}[2]{\langle #1 , #2 \rangle}

\newcommand{\K}{\mathrm{K}}

\newcommand{\opr}[1]{\mathrm{\widehat{#1}}}

\theoremstyle{plain}
\newtheorem{lemma}{Lemma}
\newtheorem{theorem}{Theorem}
\newtheorem{corollary}{Corollary}

\definecolor{BLUE}{rgb}{0.16, 0.32, 0.75}

\begin{document}
\title[Tetrads and their intrinsic spin structures]{Tetrads and their intrinsic spin structures}

\author{Tony Lyons}

\address{Department of Computing and Mathematics, South East Technological University, Waterford, Ireland}
\ead{tony.lyons@setu.ie}

\begin{abstract}
In Minkowski spacetime an observer's moving frame may be associated with a tetrad: four orthogonal directions, one timelike and three spacelike. Locally, a physical measurement is always made in relation to some tetrad, and relativistic  spin measurements in particular raise subtleties because spin and momentum cannot be cleanly separated under Lorentz boosts. This paper shows that every tetrad is endowed with an intrinsic spin structure. Each spatial direction may be generated from the timelike direction by a linear operator, and the requirement that these generate an orthogonal tetrad for \emph{every} timelike direction is sufficient, on its own, to endow them with the Lie algebra $\mathfrak{so}(1,3)$, the Pauli algebra $\ca{3,0}$ and the Lie algebra $\mathfrak{su}(2)$ simultaneously. These generators serve three roles at once: they build the spatial legs of every tetrad, they generate the Lorentz boosts between tetrads, and they endow each tetrad with an intrinsic spin structure. The eigenvectors of the spin generators form a null-ray tetrad which is simultaneously a Newman--Penrose-type frame and the set of Dirac rest spinors, while the invertibility of boosts between tetrads yields the Dirac factorisation of the mass-shell condition. Minimal coupling of gauge and spinor fields is forced rather than imposed, the potential being absorbed by the covariant derivative because the path dependence of the phase prevents it from being incorporated by the spinor phase. Wigner rotations, relativistic spin correlations, quantum reference frames and the Coulomb problem are examined within this tetrad framework. This  framework is equivalent to the Hilbert-space formalism of relativistic quantum mechanics wherever the two have been compared, and so offers a geometric formulation of that formalism rather than an alternative to it.
\end{abstract}


\vspace{2mm}
\noindent{\small\itshape This is the Accepted Manuscript version of an article accepted for publication in Physica Scripta. IOP Publishing Ltd is not responsible for any errors or omissions in this version of the manuscript or any version derived from it. This Accepted Manuscript is published under a CC BY licence.}

\maketitle

\section{Introduction}\label{sec:intro}
Tetrads are typically used in curved spacetime (see \cite{Rov2004, ZZP2006} for example) where they appear as moving frames along observer worldlines, serving as local Lorentz reference frames. In Minkowski spacetime no such locally inertial reference frame is required, since any inertial reference frame is global. Nevertheless, every inertial reference frame may be associated with a tetrad, comprising the four velocity of its associated observer together with the three spatial directions orthogonal to that timelike direction. Moreover, since tetrads may be transported in flat spacetime without alteration, they provide a useful framework by which distant observers may meaningfully compare experimental results. Operationally, a tetrad may be defined by signal exchange, with signals propagating at an invariant speed across all tetrads \cite{Ale1967, EPS2012}.

The aim of this work is to demonstrate that each such tetrad is generated by a system of linear operators applied to the four-velocity of an observer, so that the tetrad generated by any four velocity is orthogonal with respect to the Minkowski metric. Universality of these linear operators for all four velocities is sufficient to ensure they are endowed with the $\mathfrak{so}(1,3)$ Lie algebra, the $\ca{3,0}$ Clifford product, and the $\mathfrak{su}(2)$ Lie algebra. Given that the tetrad generators are in the $\mathfrak{su}(2)$ Lie algebra, this ensures the tetrads are endowed with an intrinsic spin structure, while the $\mathfrak{so}(1,3)$ character of the tetrad generators also means they serve to generate transformations between tetrads. Thus these linear operators generate orthogonal spatial directions from a timelike direction, they endow that tetrad with an intrinsic spin structure and they implement the transformations of these spin structures across all such tetrads.

Our aim then is to compare this intrinsic tetrad spin structure with the standard spin structures encountered in relativistic quantum mechanics. The ingredients of relativistic quantum mechanics -- the Pauli and spacetime Clifford algebras, the $\mathfrak{su}(2)$ spin algebra, the Dirac operator and its associated spinors -- are, in the standard development, introduced axiomatically. One begins with a Hilbert space, imposes the relevant commutation and anti-commutation relations as postulates, and then derives relativistic wave equations and the operational content of spin. Each algebraic property is stipulated separately, and their mutual consistency is then confirmed. This formalism is shown to be recoverable in the tetrad formalism, with classical Hamilton--Jacobi theory (cf. \cite{Lan1970, MS1964, Lyo2023}) providing a critical motivation for the development of relativistic wave equations from this tetrad framework.

The Clifford-algebraic treatment of Lorentz transformations developed by Zeni and Rodrigues~\cite{ZR1992}, the spacetime algebra programme of Hestenes~\cite{Hes2015}, the null-ray formalism of Penrose and Rindler~\cite{PR1984}, and the geometric-algebra-based treatments of quantum mechanics in Doran and Lasenby~\cite{DL2003}, Baylis~\cite{Bay1992}, Somaroo, Lasenby and Doran~\cite{SLD1999}, and Havel and Doran~\cite{HD2002} each show that substantial parts of relativistic quantum mechanics can be recast in geometric language, often more transparently than in the Hilbert-space setting. Baylis, Cabrera and Keselica~\cite{BCK2010} go further and argue for a geometric origin of fermion spin itself. Closest to the present construction is the eigenspinor formalism of G\"ursey~\cite{Gur1957} and Baylis~\cite{Bay1992}, in which the Lorentz transformation from a particle's rest frame to an observer's frame is itself treated as a spinor, the four velocity appearing as one of its bilinear covariants and the relation $p=mu$ taking the form of a Dirac equation. A parallel line of work, initiated by Pauli and completed by Rubinow and Keller~\cite{RK1963}, obtains the Dirac equation in the eikonal limit, with the phase satisfying the relativistic Hamilton--Jacobi equation and the spinor amplitude transported along the classical trajectories.

A separate line of work has approached these questions from the side of relativistic quantum information. The observations of Czachor~\cite{Cza1997}, Peres, Scudo and Terno~\cite{PST2002}, and Peres and Terno~\cite{PT2004} that spin and momentum cannot in general be decoupled under Lorentz boosts --- the reduced spin density matrix is not Lorentz-covariant, and the usual operational identification of spin with Stern--Gerlach measurements becomes problematic at relativistic velocities --- have motivated extensive subsequent analysis of spin--momentum entanglement~\cite{GA2002, BT2005, FBHH2010} and of the appropriate covariant spin observables~\cite{BAKG2014, SV2012}. The quantum reference frame programme of Giacomini, Castro-Ruiz and Brukner~\cite{GCB2019b} addresses this by giving spin an operational meaning as a superposition of Lorentz boosts between frames; Taillebois and Avelar~\cite{TA2021} have argued on representation-theoretic grounds that any acceptable relativistic spin observable must be intrinsic in a precise sense; and Palge and Pfeifer~\cite{PP2024} have recently given the Wigner rotation a geometric reading as a holonomy on the mass shell. These developments are complementary rather than competing, and each offers a distinct insight into the geometric content of the Hilbert-space formalism.

The central result of the paper is illustrated in Figure \ref{fig:schematic} where two tetrads $e_{a}$ and $n_{a}$ may be generated by the same set of linear maps $\Omega_{i}$ from different four velocities $e_0$ and $n_0$. These tetrads are endowed with associated spin structures, the null ray eigenvectors of $\Omega_{3}$, shown as $\kappa_{a}$ and $\lambda_{a}$. Figure \ref{fig:flow-chart} outlines the logical flow: a single condition of universality on the tetrad generators is sufficient to ensure their $\mathfrak{so}(1,3)$, $\ca{3,0}$ and $\mathfrak{su}(2)$ structures. Universality of the generators follows from the fact $k_{i}=e_0+\Omega_ie_0$ is a null ray, from which orthogonality of $e_0$ and $\Omega_ie_0$ for any four velocity $e_0$ follows. The embedding of the Pauli algebra $\ca{3,0}$ in the spacetime algebra $\ca{1,3}$ means there is a corresponding representation in this Clifford algebra also. This embedding means the invertibility of tetrad maps in the Pauli algebra translates to the Dirac factorisation of the mass shell condition in the spacetime algebra. Dirac operators then arise as a compatibility condition between tetrads. The eigenvectors of these Dirac operators are also identifiable with the null ray tetrad. 
\begin{figure}[htbp]
\centering
\begin{subfigure}[c]{0.5\textwidth}
\centering
\resizebox{\linewidth}{!}{%
\begin{tikzpicture}[
  font=\scriptsize, x=1cm, y=1cm,
  leg/.style={-{Stealth[length=2mm]}, very thick},
  gen/.style={-{Stealth[length=1.4mm]}, dashdotted, thin},
  nul/.style={nullred, densely dotted, thick},
  ar/.style={-{Stealth[length=2mm]}, thick}]

\coordinate (O) at (0,0);
\draw[nul] (O) -- ++(45:2.35)
      node[pos=0.75, sloped, above, nullred] {$\Omega_{3}\kappa_{a}=\pm\kappa_{a}$};
\draw[leg] (O) -- ++(0,2.2) node[above] {$e_{0}$};
\draw[leg] (O) -- ++(2.2,0) node[right] {$e_{i}$};
\draw[gen] ([shift=(90:0.8)]O) arc[start angle=90, end angle=0, radius=0.8];
\node at (1.12,0.42) {$\Omega_{i}$};

\coordinate (P) at (4.9,0);
\draw[nul] (P) -- ++(45:2.9)
      node[pos=0.68, sloped, above, nullred] {$\Omega_{3}\lambda_{a}=\pm\lambda_{a}$};
\draw[leg] (P) -- ++(77.62:2.3) node[above] {$n_{0}$};
\draw[leg] (P) -- ++(12.38:2.3) node[right] {$n_{i}$};
\draw[gen] ([shift=(77.62:0.8)]P) arc[start angle=77.62, end angle=12.38, radius=0.8];
\node at (6.02,0.52) {$\Omega_{i}$};

\draw[ar] ([shift=(118:2.75)]3.25,0.25) arc[start angle=118, end angle=62, radius=2.75];
\node[above] at (3.25,3.05) {$\Lambda(b)=\exp\!\left(b\,\hat{u}^{\,i}\Omega_{i}\right)$};
\end{tikzpicture}}
\caption{Two tetrads related by the boost
$\Lambda(b)$. The generator $\Omega_{i}$ maps $e_{0}$ into $e_{i}$  and $n_0$ into $n_i$. The null rays $\kappa_{a}$ and $\lambda_{a}$ are eigenvectors  of
$\Omega_{3}$. The tetrads, the null ray tetrads and their transformations are all generated by the same $\Omega_{i}$ for every four velocity.}\label{fig:schematic}
\end{subfigure}\hfill
\begin{subfigure}[c]{0.5\textwidth}
\centering
\resizebox{\linewidth}{!}{%
\begin{tikzpicture}[
  font=\small,
  box/.style={draw, rounded corners=2pt, align=center, inner sep=4pt, minimum height=8mm},
  hyp/.style={box,  very thick},
  res/.style={box},
  ar/.style={-{Stealth[length=2mm]}, thick},
  node distance=6mm and 7mm]
\node[hyp] (U) {Universal orthogonality\\ $\langle\Omega_{i}e_{0},\Omega_{j}e_{0}\rangle=-\delta_{ij}$ for all $e_{0}\in\mathcal{H}$\\ \footnotesize\S\ref{sec:universal}};
\node[res, below left=8mm and 10mm of U] (so) {$\mathfrak{so}(1,3)$ \footnotesize Lem.~\ref{lem:metric-anti-symmtry}};
\node[res, below=8mm of U] (cl) {$\ca{3,0}$ \footnotesize Thm.~\ref{thm:Clifford-product}};
\node[res, below right=8mm and 10mm of U] (su) {$\mathfrak{su}(2)$ \footnotesize Cor.~\ref{cor:su2}};
\node[res, below=8mm of so] (boost) {boosts between\\ tetrads \footnotesize\S\ref{sec:Pauli-Lorentz-reps}};
\node[res, below=8mm of cl] (legs) {spatial legs\\ $e_{i}=\Omega_{i}e_{0}$ \footnotesize\S\ref{sec:universal}};
\node[res, below=8mm of su] (spin) {null-ray tetrad\\ $\kappa_{a}$ \footnotesize\S\ref{sec:spin-ops}};
\node[res, below=9mm of boost] (wig) {Wigner rotation,\\ spin correlations\\ \footnotesize\S\ref{sec:Wigner}, \S\ref{sec:singlet}};
\node[res, below=9mm of legs] (dirac) {Dirac operator\\ $(p^{a}\gamma_{a})^{2}=m^{2}$ \footnotesize\S\ref{sec:mass-shell}};
\node[res, below=9mm of spin] (rest) {rest spinors\\ $\xi^{\pm}_{s}$ \footnotesize\S\ref{sec:spinors}};
\node[res, below=9mm of dirac, text width=52mm] (app) {intrinsic spin observables, QRF comparison, hydrogenic states \footnotesize\S\ref{sec:qrf-comparison}, \S\ref{sec:hydrogenic}};
\foreach \t in {so,cl,su} \draw[ar] (U) -- (\t);
\draw[ar] (so) -- (boost);  \draw[ar] (cl) -- (legs);  \draw[ar] (su) -- (spin);
\draw[ar] (boost) -- (wig); \draw[ar] (legs) -- (dirac); \draw[ar] (spin) -- (rest);
\draw[ar] (wig) |- (app);   \draw[ar] (dirac) -- (app); \draw[ar] (rest) |- (app);
\end{tikzpicture}}
\caption{A single universality condition on the tetrad generators forces three algebraic
structures simultaneously. These in turn yield Wigner rotations, spin correlations and Dirac currents from the spacetime algebra representation.}\label{fig:flow-chart}
\end{subfigure}
\caption{A graphical illustration of the main result of this paper and a flow chart of the logical flow of the results that follow}
\end{figure}
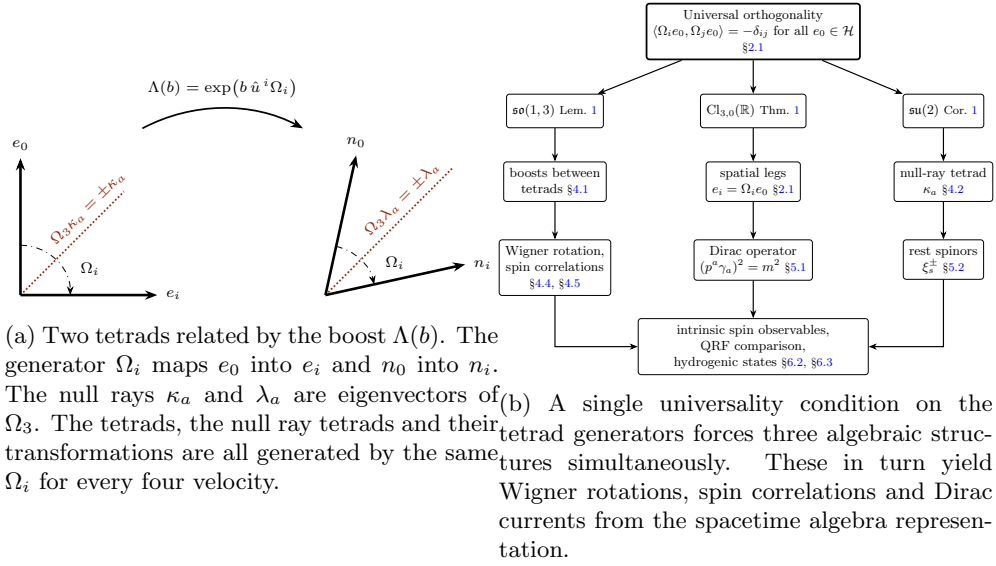
\subsection{Contribution and relation to existing work}\label{sec:novelty}
In each of the treatments above the algebra is taken as given: the Clifford
product is imposed, the basis vectors are postulated, and the geometry is read
off afterwards. What is not addressed is why the algebra should take the form it
does, nor why the same generators should govern spatial directions, boosts
between frames and spin simultaneously. If a single system of operators is to
produce an orthogonal tetrad from \emph{every} timelike direction, and not
merely from one, then no further postulate is available, and
Section~\ref{sec:tetrads} shows that this requirement alone forces all three
structures together with the identification of the imaginary unit as the
pseudoscalar.

The results believed to be new are: that universal orthogonality forces the
simultaneous $\mathfrak{so}(1,3)$, $\ca{3,0}$ and $\mathfrak{su}(2)$ structures
of the tetrad generators (\S\ref{sec:so13}--\S\ref{sec:su2}); that the
eigenvectors of the spin generators form a null ray tetrad which is
simultaneously a Newman--Penrose-type frame and the set of Dirac rest spinors
(\S\ref{sec:spin-ops}, \S\ref{sec:spinors}); that the Dirac factorisation of the
mass shell condition is a reformulation of the invertibility of a boost between
tetrads (\S\ref{sec:mass-shell}); that the spin correlations of an entangled
pair of null ray tetrads are independent of the boost relating them
(\S\ref{sec:singlet}); that minimal coupling is forced by the path dependence of
the phase when $F_{\mu\nu}\neq0$ (\S\ref{sec:min-coupling}); and that the
intrinsic spin observable of Taillebois and Avelar~\cite{TA2021} is realised
under a non-unitary spinor boost, the criterion being secured by an isometry of
the indefinite inner product rather than by unitarity
(\S\ref{sec:intrinsic}).

The advantages are of two kinds. Conceptually, the algebraic apparatus of Dirac
theory is obtained as a consequence of relativistic kinematics rather than
stipulated alongside it, and the covariant derivative acquires a geometric
origin in the non-exactness of the potential. Computationally, certain results
become correspondingly shorter: the Wigner angle is read from a composite rotor
in a single line of Clifford algebra (\S\ref{sec:Wigner}), and the hydrogenic
ground state follows from the involution $\hat{\rho}=\hat{r}^{i}\alpha_{i}$
without partial-wave decomposition or spin-angular functions
(\S\ref{sec:hydrogenic}). By contrast the Clifford and spacetime algebras
themselves, the Lorentz rotor of~\cite{ZR1992}, the eigenspinor
of~\cite{Gur1957, Bay1992}, the Dirac equation, the Wigner rotation, the eikonal
expansion of~\cite{RK1963} and the quantum-reference-frame construction
of~\cite{GCB2019b} are established results; what is offered is a derivation of
their common algebraic origin, not a replacement for them.

\section{Orthogonal tetrads}\label{sec:tetrads}
\subsection{Universal tetrad generators}\label{sec:universal}
A tetrad $\{e_{a}\}_{a=0}^{3}$ is a collection of vectors in Minkowski spacetime $\mathbb{R}^{1,3}$ which are orthogonal with respect to the bilinear form $\langle\cdot,\cdot\rangle$ where
\begin{equation}
\bl{e_{a}}{e_{b}} = e_{a}^{\top}\eta e_{b} = \eta_{ab}\qquad \eta=\diag(1,-1,-1,-1).
\end{equation}
The timelike vector $e_{0}$ of any tetrad belongs to the unit hyperboloid 
\[\mathcal{H}=\left\{e_0\in\mathbb{R}^{1,3}\;\middle|\; \bl{e_0}{e_0}=1\right\}\]
while the orthogonal complement of any $e_{0}\in\mathcal{H}$, namely
\begin{equation}
e_{0}^{\perp}\coloneq \left\{v\in\mathbb{R}^{1,3}\; \middle|\: \bl{e_0}{v}=0\right\}
\end{equation}
is spanned by associated spacelike vectors $\{e_{i}\}_{i=1}^{3}$. Operationally, a timelike vector $e_0$ is the four velocity of an observer/particle while the spacelike directions $e_{i}$ arise from signal propagation along null rays  
\begin{equation}
k_{i} \coloneq e_i + e_0,
\end{equation} 
such that $\bl{k_i}{k_i}=0$ for $i=1,2,3$ (see \cite{Ale1967, EPS2012} for example). 

It is an elementary fact of linear algebra that given the orthogonal tetrad $\{e_{a}\}_{a=0}^{3}$, there are linear operators $\{\Omega_{i}\}_{i=1}^{3}$ such that 
\begin{equation}
e_i = \Omega_{i}e_{0},\quad i=1,2,3.
\end{equation} 
In fact for any pair $\{e_{0}$, $e_{i}\}$ there are an infinite number of such linear operators. The main result of this section demonstrates that there are linear operators $\Omega_{i}$ such that 
\begin{equation}\label{eq:ortho}
\bl{\Omega_i e_0}{e_0} = 0\quad \bl{\Omega_i e_0}{\Omega_j e_0}=-\delta_{ij} \qquad \forall \ e_{0}\in\mathcal{H}.
\end{equation}
That is to say the tetrad generators $\Omega_{i}$ are not simply restricted to a specific timelike vector $e_{0}$ but rather may be used to generate an orthogonal tetrad from any unit timelike vector $e_{0}\in\mathcal{H}$.

Physically, we may consider two labs with associated tetrads $e_{a}$ and $n_{a}$ initially at rest in relation to each other in which case $e_0=n_0$. As such, the orthogonal complement of each timelike vector may be generated by a common set of linear operators so that 
$$
\langle e_0,\Omega_{i}e_0 \rangle=0 \qquad \langle n_0,\Omega_{i}n_0 \rangle=0. 
$$ 
Since $e_{i}$ and $n_{i}$ are operationally defined by signal propagation with $n_i=e_i=k_i-e_0$ we now find
\[(\mathbf{1}+\Omega_i)n_0 = (\mathbf{1}+\Omega_i)e_0=k_i \] 
are null,  and so it is required that 
\[ \langle (\mathbf{1}+\Omega_i) e_0 , (\mathbf{1}+\Omega_i) e_0 \rangle = 0\]
and likewise for $n_0$. 

Let us now consider the case where the lab frames are in uniform relative motion so that $n_0=\Lambda e_0$. The principle of relativity ensures $(\mathbf{1}+\Omega_{i})e_0=k_i$ is null in the rest frame coordinates of $e_{a}$. Likewise in the rest frame coordinates of $n_a$ it is also required that $(\mathbf{1}+\Omega_{i})n_0=k_i$ is null. However, in the rest frame coordinates of $e_{a}$ we now have $(\mathbf{1}+\Omega_{i})n_{0}=(\mathbf{1}+\Omega_{i})\Lambda e_{0}$. Since the nullness of the operational vector  $(\mathbf{1}+\Omega_{i})n_{0}$ is a property of the vector itself and not a coordinate-system artefact, it follows that
\[ \langle (\mathbf{1}+\Omega_{i})\Lambda e_{0} , (\mathbf{1}+\Omega_{i})\Lambda e_{0} \rangle = 0. \] 
Hence $\Omega_{i}$ may be used to generate an orthogonal  tetrad from any $n_{0}\in\mathcal{H}$. Moreover, since $\Lambda^{\top}\eta=\eta\Lambda^{-1}$ it is also found that
\begin{equation}\label{eq:new_gens}
\langle (\mathbf{1}+\Omega_i) \Lambda e_0 , (\mathbf{1}+\Omega_i) \Lambda e_0 \rangle 
= \langle (\mathbf{1}+\Lambda^{-1}\Omega_i\Lambda) e_0 , (\mathbf{1}+\Lambda^{-1}\Omega_i\Lambda) e_0 \rangle = 0
\end{equation}
so that $\Lambda^{-1}\Omega_{i}\Lambda$ is another set of tetrad generators.

This universality of the tetrad generators $\Omega_{i}$ is sufficient to ensure they furnish a representation of the Lie algebra $\mathfrak{so}(1,3)$ (Lemma \ref{lem:metric-anti-symmtry}), Pauli algebra $\cl_{3,0}(\mathbb{R})$ (Theorem \ref{thm:Clifford-product}) and the Lie algebra $\mathfrak{su}(2)$ (Corollary \ref{cor:su2}). The crucial condition here is universality: the orthogonality condition \eqref{eq:ortho} must hold for every $e_{0}\in\mathcal{H}$, not merely for one. This single condition is what forces the algebraic structure in the three results that follow, with no further input.

\subsection{The $\mathfrak{so}(1,3)$ Lie algebra }\label{sec:so13}
\begin{lemma}[Metric anti-symmetry of tetrad generators]\label{lem:metric-anti-symmtry}
If 
\[\bl{\Omega_{i}e_0}{e_0} = \bl{e_0}{\Omega_{i}e_0}=0\]
for all $e_{0}\in\mathcal{H}$ then
\[\Omega_{i}^{\top}\eta+\eta\Omega_{i}=0.\]
\end{lemma}

\begin{proof}
The set of timelike vectors
\begin{equation}
\mathcal{T}\coloneq \left\{v\in\mathbb{R}^{1,3}\;\middle|\; \bl{v}{v}>0 \right\}
\end{equation}
is open in $\mathbb{R}^{1,3}$ with respect to the usual Euclidean topology on $\mathbb{R}^{4}$. 
Given any vector $v\in\mathcal{T}$ there is a corresponding vector $e_{0}\in\mathcal{H}$ with 
\[e_{0}=\frac{1}{\sqrt{\bl{v}{v}}}v \qquad \bl{e_{0}}{e_{0}}=1.\]
Since $\bl{\Omega_{i}e_{0}}{e_0}=\bl{e_0}{\Omega_{i}e_{0}}=0$ for all $e_{0}\in\mathcal{H}$, it now follows that 
\begin{equation}
\bl{\Omega_{i}e_{0}}{e_0}+\bl{e_0}{\Omega_{i}e_{0}}=0 \implies \frac{1}{\bl{v}{v}}\left(\bl{\Omega_{i}v}{v}+\bl{v}{\Omega_{i}v}\right) = 0.
\end{equation}
It follows that the symmetric quadratic form
$$
Q_{i}(v)\coloneq \bl{\Omega_{i}v}{v}+\bl{v}{\Omega_{i}v}
$$
is vanishing for all $v\in \mathcal{T}$. The polarisation identity ensures the quadratic form $Q_{i}(\cdot)$ may be associated with a symmetric bi-linear form $B_{i}(\cdot,\cdot)$ according to
\begin{equation}
B_{i}(u , v) \coloneq \frac{1}{2}\left(Q_{i}(u+v)-Q_{i}(u)-Q_{i}(v)\right)\quad \forall\ u,v\in\mathcal{T},
\end{equation}
see ref. \cite{HK1971} for example.  
Moreover, since $\mathcal{T}$ is open in $\mathbb{R}^{4}$, then for any $u,v\in\mathcal{T}$ there is some  $\delta>0$ such that $u+tv\in\mathcal{T}$ for all $\abs{t}<\delta$. This now means that   
\begin{equation}
Q_{i}(u+tv) = Q_{i}(u) + 2tB_{i}(u,v)+t^2Q_{i}(v)=0 \implies tB_{i}(u,v)=0 
\end{equation}
for all $u,v\in\mathcal{T}$ and $\abs{t}<\delta$.  

As $\mathcal{T}$ is open in $\mathbb{R}^{4}$ and non-empty, it contains an open ball, and any open ball in $\mathbb{R}^{4}$ contains four linearly independent vectors. Hence $\mathcal{T}$ contains vectors $\{v_{a}\}_{a=0}^3$ which span $\mathbb{R}^{1,3}$. Since $B_{i}(u,v)=0$ for all $u,v\in\mathcal{T}$ it now follows that
\begin{equation}
B_{i}(v_{a},v_{b}) = v_{a}^{\top}\left(\Omega_{i}^{\top}\eta +\eta \Omega_{i}\right)v_{b} = 0.
\end{equation}
Given that $\spn\{v_{a}\}_{a=0}^{3}=\mathbb{R}^{1,3}$ it follows that 
$$\Omega_{i}^{\top}\eta +\eta \Omega_{i}=0.$$
Hence $\Omega_{i}\in\mathfrak{so}(1,3)$ and so $\{\Omega_{i}\}_{i=1}^{3}$ are generators of the Lorentz group (see \cite{Wei1995} for example).  
\end{proof}

\subsection{The $\ca{3,0}$ Clifford algebra}\label{sec:Cl30}
The main result of this section is to now show that these linear operators $\Omega_{i}$ are necessarily endowed with  the Clifford product
\begin{equation}\label{eq:CP-Pauli}
\Omega_{i}\Omega_{j} + \Omega_{j}\Omega_{i}=2\delta_{ij}\mathbf{1},\qquad \mathbf{1}=\diag(1,1,1,1)
\end{equation}
and so furnish a representation of the Pauli algebra $\ca{3,0}$, see \cite{Lou2001} for example.

\begin{theorem}[Clifford products for tetrad generators]\label{thm:Clifford-product}
If the vectors  $e_{0}$ and $e_{i}=\Omega_{i}e_{0}$ form an orthogonal tetrad with respect to the bi-linear form $\bl{\cdot}{\cdot}$ for all $e_{0}\in\mathcal{H}$, such that 
\begin{equation}
\bl{e_0}{\Omega_{i}e_0}=0\quad \bl{\Omega_{i}e_0}{\Omega_{j}e_0}=-\delta_{ij},
\end{equation}
then the linear operators $\Omega_{i}$ are endowed with the Clifford product
\begin{equation}
\Omega_{i}\Omega_{j}+\Omega_{j}\Omega_{i}=2\delta_{ij}\mathbf{1}.
\end{equation}
\end{theorem} 

\begin{proof}
For any $v\in\mathcal{T}$ there is a corresponding  $e_{0}=\bl{v}{v}^{-1/2}v\in\mathcal{H}$, the orthogonality condition $\bl{\Omega_{i}e_{0}}{\Omega_{j}e_{0}}=-\delta_{ij}$ ensures 
\begin{equation}\label{expr:vivj-1}
\bl{\Omega_{i}v}{\Omega_{j}v} + \bl{\Omega_{j}v}{\Omega_{i}v}=-2\delta_{ij}\bl{v}{v}.
\end{equation}
 Lemma \ref{lem:metric-anti-symmtry} ensures $\Omega_{i}^{\top}\eta+\eta\Omega_{i}=0$ and so the left hand side may be reorganised as 
\begin{equation}\label{expr:vivj-2}
\bl{\Omega_{i}v}{\Omega_{j}v} + \bl{\Omega_{j}v}{\Omega_{i}v} = -\bl{v}{v}e_{0}^{\top}\eta(\Omega_{i}\Omega_{j} + \Omega_{j}\Omega_{i})e_{0} = -\bl{v}{(\Omega_{i}\Omega_{j}+\Omega_{j}\Omega_{i}) v}
\end{equation}
Equating \eqref{expr:vivj-1} and \eqref{expr:vivj-2}, it is found  that the quadratic form
\begin{equation}
Q_{ij}(v) = \bl{v}{(\Omega_{i}\Omega_{j}+\Omega_{j}\Omega_{i} -2\delta_{ij}\mathbf{1})v}=0,
\end{equation}
for all $v\in\mathcal{T}$. 

 The polarisation identity and open-set arguments  applied in Lemma \ref{lem:metric-anti-symmtry} once again ensure the associated bilinear form $\Omega_{i}\Omega_{j}+\Omega_{j}\Omega_{i} -2\delta_{ij}\mathbf{1}$ vanishes. This ensures the tetrad generators $\Omega_{i}$ are endowed with the Clifford product
\begin{equation}
\Omega_{i}\Omega_{j}+\Omega_{j}\Omega_{i}=2\delta_{ij}\mathbf{1}
\end{equation}
so the tetrad generators $\{\Omega_{i}\}_{i=1}^{3}$ are endowed with the Clifford product of the Pauli algebra $\cl_{3,0}(\mathbb{R})$.
\end{proof}
In contrast to the geometric algebra formalism  (see \cite{ZR1992, Hes2015, Lou2001} for example) where abstract Clifford algebras are postulated a priori, the tetrad generators here are matrices endowed with the corresponding Clifford product. This Clifford algebra structure follows directly from the requirement that the matrices generate tetrad legs orthogonal to any four velocity. The novelty of the result here is the  geometric framework provided for the Clifford product \eqref{eq:CP-Pauli} which is typically introduced axiomatically defining the abstract Pauli algebra $\cl_{3,0}(\mathbb{R})$.

\subsection{The $\mathfrak{su}(2)$ Lie algebra}\label{sec:su2}
Orthogonality of the tetrads $\Omega_{i}e_{0}$ for all $e_{0}\in\mathcal{H}$ now ensures $\Omega_{i}\in\ca{3,0}$, which in turn means that $\Omega_{i}$ furnish a representation of the  Lie algebra $\mathfrak{su}(2)$. 
\begin{corollary}[The $\mathfrak{su}(2)$ algebra from tetrad generators]\label{cor:su2}
The linear operators $\{\Omega_{i}\}_{i=1}^{3}$ endowed with the Clifford product  \eqref{eq:CP-Pauli} are also endowed with the $\mathfrak{su}(2)$ Lie algebra
\begin{equation}
\Omega_{i}\Omega_{j} - \Omega_{j}\Omega_{i}=2\iu\varepsilon_{ijk}\Omega_{k}
\end{equation}
where $\varepsilon_{ijk}$ is the Levi-Civita symbol in $\mathbb{R}^{3}$. 
\end{corollary}

\begin{proof}
The Clifford product \eqref{eq:CP-Pauli} now ensures  $\Omega_{i}\Omega_{j}=-\Omega_{j}\Omega_{i}$ when $i\neq j$. This in turn means that $(\Omega_{1}\Omega_{2}\Omega_{3})^{2}=-\mathbf{1}$. Since $\Omega_{i}$ are matrices, this means that the pseudoscalar $\mathrm{I}=\Omega_{1}\Omega_{2}\Omega_{3}$ may be identified with the complex unit matrix $\iu\mathbf{1}$  where $\iu$ is the imaginary unit and 
\begin{equation}
\mathrm{I} = \Omega_{1}\Omega_{2}\Omega_{3}=\iu\mathbf{1}.\label{eq:pseudo}
\end{equation}   
Since $\Omega_{3}^{2}=\mathbf{1}$, it now follows that 
\begin{equation}
\Omega_{1}\Omega_{2} = \mathrm{I}\Omega_{3}=\iu\Omega_{3}.
\end{equation}
Moreover, the Clifford product \eqref{eq:CP-Pauli} ensures $\Omega_{2}\Omega_{1}=-\Omega_{1}\Omega_{2}$ to give
\[\comm{\Omega_{1}}{\Omega_{2}}=2\iu \Omega_{3}.\]
Similar identities are attained under cyclic permutations so that we deduce
\begin{equation}\label{eq:Lie-su2}
\Omega_{i}\Omega_{j}-\Omega_{j}\Omega_{i}=2\iu\varepsilon_{ijk}\Omega_{k},
\end{equation} 
the $\mathfrak{su}(2)$ Lie algebra. 
\end{proof}
In the geometric algebra framework, the pseudoscalar is an abstract algebra member which is central and whose negative square is the identity element. The geometric framework ensures $\Omega_{i}\in\mathrm{GL}(4;\mathbb{C})$ so that the pseudo-scalar may be identified with the imaginary matrix $\mathrm{I}=\iu\mathbf{1}$. The imaginary unit typically appears via the operator-momentum identification $p^{\mu} \leftrightarrow \iu\hbar \partial^{\mu} $. In the tetrad framework it now follows directly from the orthogonality of the tetrad legs  $e_i=\Omega_{i}e_0$ generated by each four velocity $e_{0}\in\mathcal{H}$.

\section{Kinematics and phase waves}\label{sec:kinetic-phase}

\subsection{Integral curves}\label{sec:curves}
In coordinates $y^{\mu}$ adapted to the tetrad $n_{a}$, the integral curves of $n_{0}^{\mu}=\delta_{0}^{\mu}$ are governed by $\dv{y}{\tau}=n_{0}(y(\tau))$, yielding the worldline congruence
\begin{equation}\label{eq:ic-y}
y(\tau) = (\tau,\vec{y}), \qquad \vec{y}\in\mathbb{R}^{3}.
\end{equation}
Without loss of generality, the particle associated with $n_{a}$ follows $(\tau,\vec{0})$ with four-velocity purely along $n_0$. In coordinates $x^{\mu}$ adapted to the tetrad $e_{a}$, this congruence is
\begin{equation}\label{eq:cong-x}
\dv{x}{\tau} = n_{0}(x(\tau))=\left(u^{0}(x(\tau)) + u^{i}(x(\tau))\Omega_{i}\right)e_{0}.
\end{equation}
The congruence corresponds to level sets of a scalar field $\tau(x)$, mapping $e_a$-adapted coordinates $x^\mu$ to the proper time $\tau$ of $n_a$.
For $\partial\tau(x) \in\mathcal{H}$, normalization requires
\begin{equation}\label{eq:ge1} 
\partial_{\mu}\tau(x)\partial^{\mu}\tau(x) = 1. 
\end{equation}
The associated four-momentum has components $p^\mu = m n_0^\mu = m\partial^\mu\tau(x)$ and satisfies $p_\mu p^\mu = m^2$. The continuity equation $\partial_\mu p^\mu = 0$ may be re-formulated as 
\begin{equation}\label{eq:ge2}
\partial_{\mu}\partial^{\mu}\tau(x) = 0.
\end{equation}

\subsection{Acceleration and boost re-parameterisations}
Observers in uniform relative motion possess tetrads generated by the boost operators $\Omega_i$. To treat accelerated motion symmetrically—allowing either observer to be considered inertial—one must resolve proper-time discontinuities at scattering events. 
\begin{figure}[h!]
\centering
\begin{subfigure}{0.45\textwidth}
\centering
\includegraphics[width=0.99\textwidth]{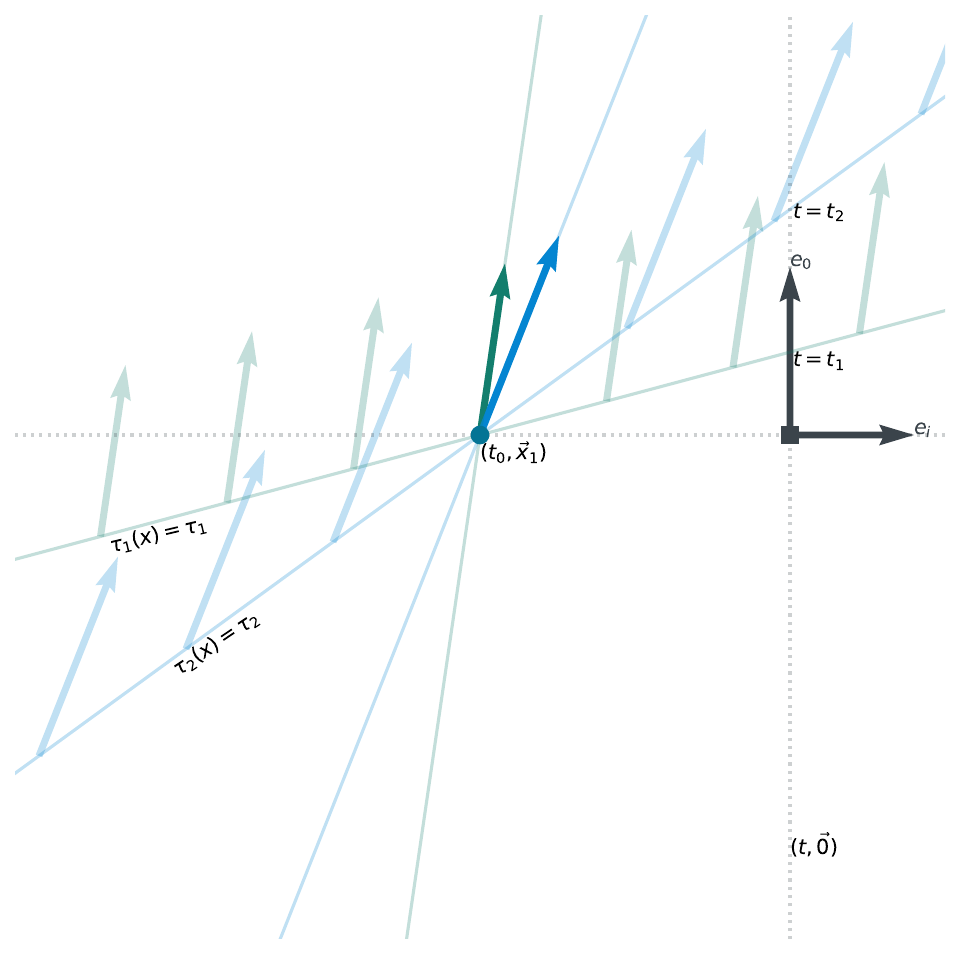}
\caption{Acceleration of $\mathscr{B}$ in the frame $e_{a}$ of $\mathscr{A}$}
\label{fig:acc_in_ea}
\end{subfigure}\hfil
\begin{subfigure}{0.45\textwidth}
\centering
\includegraphics[width=0.99\textwidth]{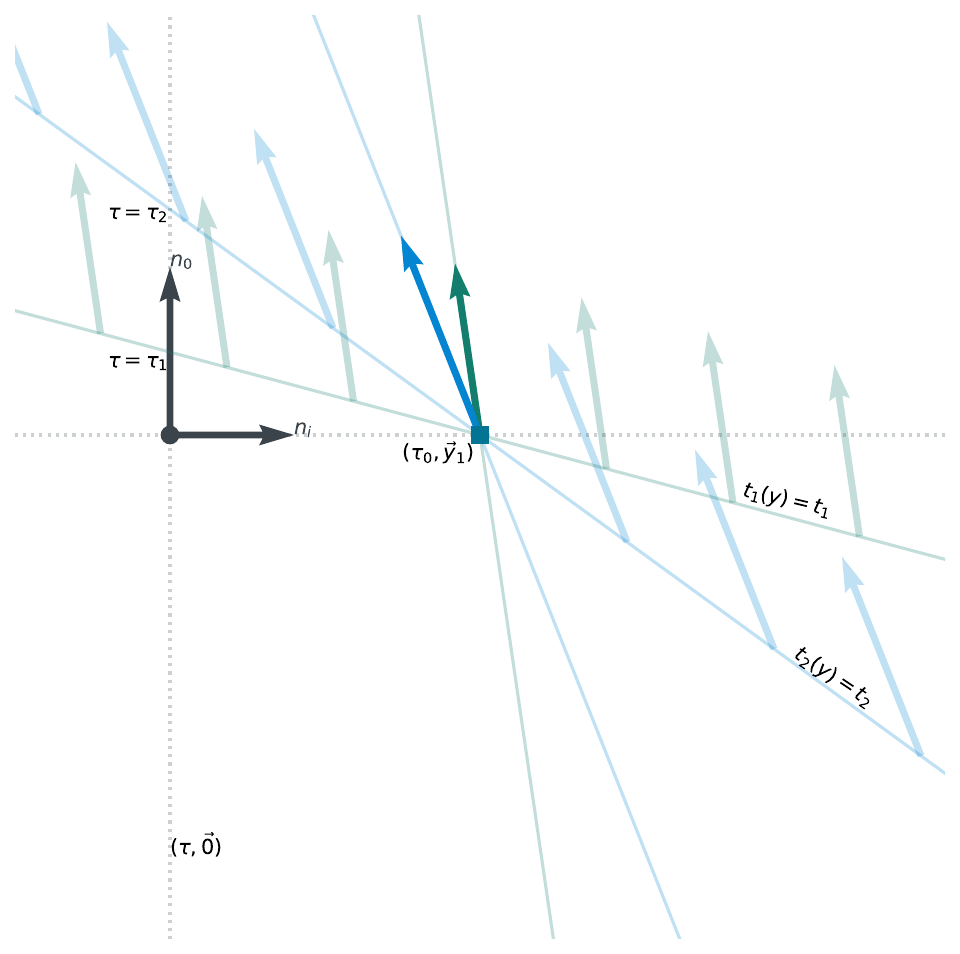}
\caption{Acceleration of $\mathscr{A}$ in the frame $n_{a}$ of $\mathscr{B}$}
\label{fig:acc_in_na}
\end{subfigure}
\caption{Observer scattering described in mutual moving frames.}
\end{figure}
Figures~\ref{fig:acc_in_ea}--\ref{fig:acc_in_na} depict observers $\mathscr{A}$ ($\blacksquare$) and $\mathscr{B}$ ($\bullet$). In $e_a$-adapted coordinates, $\mathscr{B}$ scatters at $x_S = (t_0, \vec{x}_1)$, switching four-velocity gradients from $\partial^a\tau_1(x)$ to $\partial^a\tau_2(x)$. Standard linear potentials $\tau_1(x)=u_\mu x^\mu$ and $\tau_2(x)=v_\mu x^\mu$ yield $u_0 t_0 - \vec{u}\cdot\vec{x}_1 \neq v_0 t_0 - \vec{v}\cdot\vec{x}_1$, causing an apparent discontinuity $\tau_1(x_S) \neq \tau_2(x_S)$ (and vice-versa in frame $n_a$).  To eliminate this artifact and preserve continuous individual proper times, the classical velocity potential $\tau(x)=q_\mu x^\mu$ must be generalised to a potential $\phi(x)$ such that $\tau(x) = g(\phi(x))$ obeys \eqref{eq:ge1}--\eqref{eq:ge2}.

\subsection{Velocity potentials from de Broglie waves}\label{sec:vel-pot}
Level-set matching of $\tau(x)$ and $\phi(x)$ requires invertible maps $\tau(x)=g(\phi(x))$ and $\phi(x)=f(\tau(x))$. Applying the $\tau \to -\tau$ discrete symmetry of the system \eqref{eq:ge1}--\eqref{eq:ge2}  selects $g(\phi) = K\ln(\phi)$ (see \ref{apx:vp}), giving
\begin{equation}
\phi(x) = \mathrm{e}^{K\tau} = \mathrm{e}^{K u_{\mu}x^{\mu}}, \qquad u\in\mathcal{H}.
\end{equation}
Real $K$ preserves the proper-time discontinuity under re-parameterization due to monotonicity of the real exponential function. However, under the imaginary parameterisation $K=i\kappa$ ($\kappa\in\mathbb{R}$) ensures single-valued continuity of the potential $\phi$ across scattering when
\begin{equation}
\kappa (u_{0}t-\vec{u}\cdot\vec{x}) - \kappa (v_{0}t-\vec{v}\cdot\vec{x}) = 2n\pi, \qquad n\in\mathbb{Z}.
\end{equation}
For $\kappa = \pm m/\hbar$, $\phi(x) = \mathrm{e}^{\pm i p_\mu x^\mu / \hbar}$ transforms the governing equations \eqref{eq:ge1}--\eqref{eq:ge2} into the Klein-Gordon equation,
\begin{equation}\label{eq:KG}
\partial^{\mu}\partial_{\mu}\phi + \frac{m^2}{\hbar^2}\phi = 0,
\end{equation}
identifying $\phi$ as a de Broglie phase wave in sync with rest-frame clock frequencies. In Figure~\ref{fig:phase_scatter}, the transition in $\phi$ across a scattering event occurs simultaneously everywhere. Since $\phi$ is not a locally emitted wave, but rather the velocity potential encoding the global frame transformation between $e_a$ and $n_a$, this global shift involves no information transfer and respects locality.
\begin{figure}[h!]
  \centering
  \begin{subfigure}[b]{0.32\textwidth}
    \includegraphics[width=\textwidth]{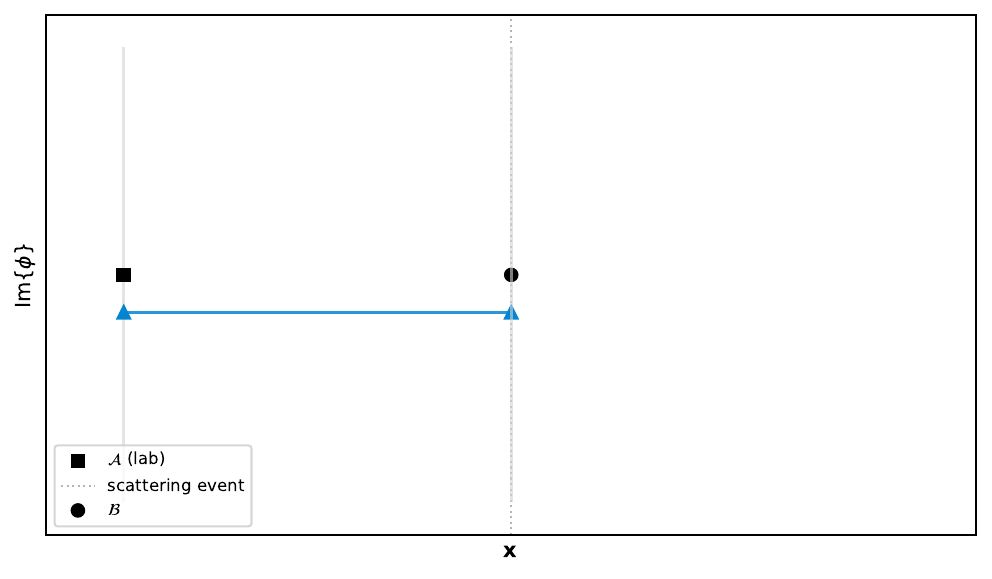}
    \caption{$t = t_0 - \varepsilon$}
    \label{fig:scatter_tm}
  \end{subfigure}
  \begin{subfigure}[b]{0.32\textwidth}
    \includegraphics[width=\textwidth]{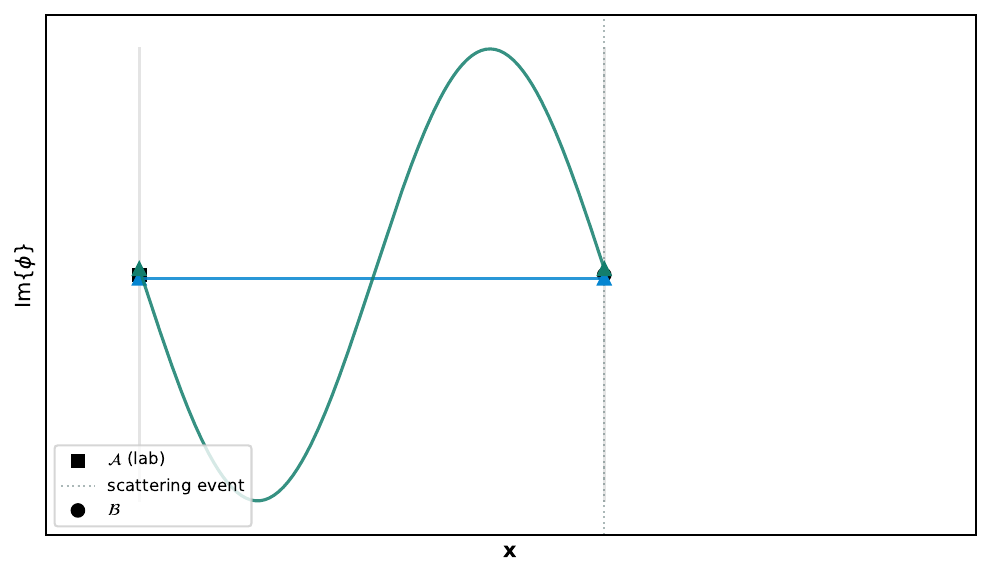}
    \caption{$t = t_0$}
    \label{fig:scatter_t0}
  \end{subfigure}
  \begin{subfigure}[b]{0.32\textwidth}
    \includegraphics[width=\textwidth]{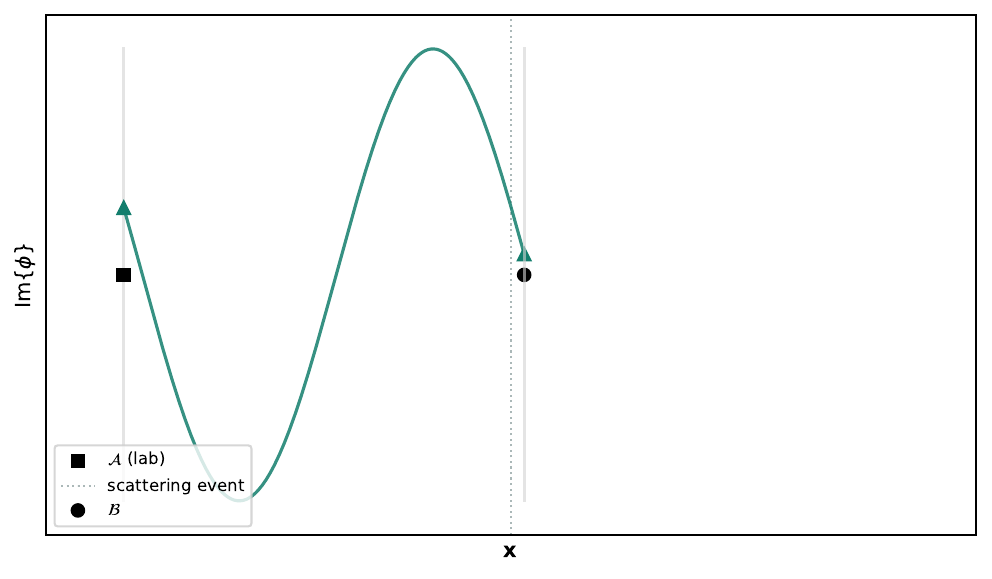}
    \caption{$t = t_0 + \varepsilon$}
    \label{fig:scatter_tp}
  \end{subfigure}
  \caption{Snapshots of $\Im\{\phi(x)\}$ in $e_{a}$ during the scattering of $\mathscr{B}$ at $(t_0,\vec{x}_1)$. Phase matching at $t_0$ prevents proper-time discontinuities for both observers.}
  \label{fig:phase_scatter}
\end{figure}
Thus the governing system \eqref{eq:ge1}--\eqref{eq:ge2} is formally equivalent to the Klein-Gordon equation \eqref{eq:KG} for a single orbit wave function, as previously established in \cite{MS1964}. In the current framework the phase $\phi$ has an additional function as a velocity potential and so has a direct geometric significance relating integral curves in one set of tetrad adapted coordinates to another tetrad coordinate system.
Earlier works have connected Hamilton--Jacobi theory and de Broglie waves \cite{deB1925, Dir1933, Syn1960, McK1978, But2004}, but restricted their focus to particle trajectories within a single, fixed reference frame or background observer. In contrast, the present work extends $\phi(x)$ from a single-frame trajectory guide to a global velocity potential encoding the coordinate transformations between mutually accelerating tetrad frames. This proposal was first developed in \cite{Lyo2023} in the context 1+1 dimensional spacetime. The imaginary parameter $K=i\kappa$ enforces single-valued phase matching across scattering events, and demonstrates how de Broglie phase waves resolve apparent proper-time discontinuities between accelerated observers without violating relativistic locality.

\subsection{Acceleration via potential fields}\label{sec:potentials}
A charged particle ($m, q$) moving in a background potential $A_{\mu}(x)$ follows an integral curve $x(\tau,\vb{y})$. The action governing this trajectory (and thus the tetrad mapping) is
\begin{equation}\label{eq:action}
I[x(\tau,\vb{y})] = \int_{\tau_{1}}^{\tau_{2}} \left[m\sqrt{\dot{x}_{\mu}\dot{x}^{\mu}} + q A_{\mu}(x)\dot{x}^{\mu}\right]\dd{\tau}, \qquad \dot{x}^\mu \equiv \frac{\dd x^\mu}{\dd \tau},
\end{equation}
see \cite{LL2013} for example. 
Evaluating this yields the Hamilton--Jacobi and continuity equations
\begin{align}
(\partial^{a}S-qA^{a})(\partial_{a}S-qA_{a}) &= m^2, \label{eq:HJ}\\
\partial_{a}(\partial^{a}S-qA^{a}) &= 0, \label{eq:contin}
\end{align}
where kinetic momentum $p^{a} = m\dot{x}^{a} = \partial^{a}S - q A^{a} \in \mathcal{H}$ parameterises an invertible tetrad map.

Along an integral curve, introducing the phase shift $\theta(x) = \frac{q}{m} \int_{\tau_{1}}^{\tau} A_{a}\dot{x}^{a}\dd{\tau'}$ keeps action \eqref{eq:action} and field system \eqref{eq:HJ}--\eqref{eq:contin} covariant with $\dot{\theta} = A_a \dot{x}^a$. Transverse derivatives across the congruence yield
\begin{equation}
\pdv{\theta}{y^{i}} = \frac{q}{m}\left[A_{\mu}J^{\mu}_{\ i}\right]_{\tau_1}^{\tau} + \frac{q}{m}\int_{\tau_{1}}^{\tau} F_{\rho\mu}\dot{x}^{\mu}J^{\rho}_{\ i}\,\dd{\tau'},
\end{equation}
where $J^\mu_{\ i} = \partial x^\mu / \partial y^i$ and $F_{\mu\nu} = \partial_\mu A_\nu - \partial_\nu A_\mu$. The non-vanishing field strength $F_{\rho\mu}$ makes $\theta(x)$ intrinsically path-dependent \cite{JO2001}, preventing the global absorption of $A_a$ into $S$ via a gauge shift unless $F_{\mu\nu}=0$.

\section{Boosts of tetrad spin structures}
Section \ref{sec:tetrads} has established that the requirement to generate an orthogonal tetrad from every four velocity $e_{0}\in\mathcal{H}$ forces the simultaneous $\mathfrak{so}(1,3)$, $\ca{3,0}$ and $\mathfrak{su}(2)$ structures of the operators $\Omega_{i}$. The tetrad generators consequently fulfil three interrelated functions at once: they generate the spatial legs $e_{i}=\Omega_{i}e_{0}$ of each tetrad, they generate the boosts $n_{0}=\Lambda e_{0}$ between tetrads, and they endow every tetrad with an intrinsic spin structure which those same boosts carry between tetrads. This section develops the transformation properties of these spin structures and their consequences for relativistic spin measurements.

\subsection{Tetrad transformations}\label{sec:Pauli-Lorentz-reps}
Since $\Omega_{i}\in\mathfrak{so}(1,3)$ by Lemma \ref{lem:metric-anti-symmtry}, the transformation between the lab frame $e_{a}$ and the rest frame $n_{a}$ is generated by the tetrad generators themselves. Expanding the particle four velocity $n_{0}$ in the lab frame,
\begin{equation}\label{eq:jac-kin-rep}
n_{0} = u^{0}e_{0} + u^{i}e_{i} = \left(\cosh(b)\mathbf{1}+ \sinh(b)\opr{U}\right)e_{0}=\mathrm{e}^{b \opr{U}} e_{0},
\end{equation}
cf. \cite{ZR1992}, where $\opr{U}=\hat{u}^{i}\Omega_{i}$ satisfies $\opr{U}^{2}=\mathbf{1}$ for any $\mathbf{\hat{u}}\in S^{2}$, and $u^{a}$ are the four velocity components of the particle with respect to the lab frame.

The Clifford algebra generators $\Omega_{i}\in\ca{3,0}$ provide a particularly direct route from the Lie algebra $\mathfrak{so}(1,3)$ and the Lie group $\mathrm{SO}(1,3)$ since Zeni and Rodrigues~\cite{ZR1992} establish 
\begin{equation}\label{eq:Lorentz-gen}
\Lambda(b) = e^{b \opr{U}} = \cosh(b)\mathbf{1}+\sinh(b)\opr{U} .
\end{equation}
The metric anti-symmetry condition established in Lemma \ref{lem:metric-anti-symmtry} ensures $\eta\opr{U}^{\top}\eta=-\opr{U}$, while the Clifford product \eqref{eq:CP-Pauli} ensures
\begin{equation}
\eta\Lambda^{\top}(b)\eta=\cosh(b)\mathbf{1}-\sinh(b)\opr{U}=\Lambda(-b).
\end{equation}
While equation \eqref{eq:Lorentz-gen} is well-established in many group theoretic treatments of the Lorentz group (see \cite{Wei1995, Car2000} for example), a particular advantage of the Clifford algebra treatment studied by Zeni and Rodrigues in \cite{ZR1992} is that this is an exact result.

Since $\Omega_{i}\in\mathfrak{so}(1,3)$ and $\Omega_{i}\in\mathfrak{su}(2)$,
the vectors $u=\mathrm{e}^{b\opr{U}}e_{0}$ and $v=\mathrm{e}^{\beta\hat{V}}e_{0}$ ($\hat{V}=\hat{v}^{i}\Omega_{i}$)
belong to the representation spaces of both $\mathrm{SO}(1,3)$ and
$\mathrm{SU}(2)$, and so are endowed with two distinct inner products
\[\langle u , v \rangle = u^{\top}\eta v  \qquad \ip{u}{v}=u^{\dag}v.\]
In tetrad adapted coordinates where $e^{\mu}_{a}=\delta^{\mu}_{a}$ these inner products take the form
\[\langle e_{a} , e_{b}\rangle = \eta_{ab}\qquad \ip{e_{a}}{e_{b}} = \delta_{ab}.\]
The boosted tetrad $n_i=\Omega_{i}n_{0}  = \Omega_{i}\Lambda(b)e_{0}$ remains orthogonal with respect to the Minkowski inner product with $\langle n_{a} , n_b \rangle = \eta_{ab}$, where
\begin{equation}
n_{i}=\Omega_{i}(\cosh(b)+\sinh(b)\hat{u}^{j}\Omega_{j})e_{0} =(\hat{u}^{j}\sinh(b)\delta_{ij} + \cosh(b)\Omega_{i}  + \iu\sinh(b)\varepsilon_{ijk}\hat{u}^{j}\Omega_{k})e_0.
\end{equation}
We also note that $\Lambda(b)e_{i}=n_{i}^{*}$ where $n_{i}$ and $n_{i}^{*}$ are complex conjugates. Moreover, the pseudoscalar defined by $(\Omega_1\Omega_{2}\Omega_{3})^2=-\mathbf{1}$ (cf. Corollary \ref{cor:su2}) has a complementary solution $\Omega_{i}^{*}$ related to $\Omega_{i}$ by complex conjugation. Moreover, in the canonical representation these commute:
\[\comm{\Omega_{i}}{\Omega_{j}^{*}}=0\quad i,j=1,2,3\]
so that $\Omega_{i}$ and $\Omega_{j}^{*}$ generate two independent chiral sectors (see \cite{Wei1995} for example). As such, $n_{i}^{*}=\Omega_{i}^{*}n_{0}=\Lambda(b)e_{i}$ while $n_{i}=\Omega_{i}\Lambda(b)e_{0}$.  

\begin{figure}[h!]
\centering
\includegraphics[width=0.5\textwidth,trim=0cm 0cm 0cm 0cm,clip=true]{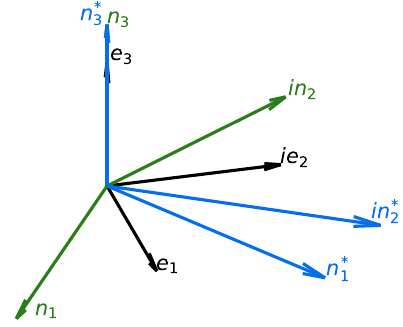}
\caption{The spatial legs of the tetrads $e_{a}$, $n_{a}$ and $n_{a}^{*}$ in the rest
frame coordinates of $e_{a}$, where the tetrads are related by the boost
$\Lambda(0.75)=\mathrm{e}^{0.75\,\Omega_{3}}$. Note that $n_{0}=n_{0}^{*}$ and $n_{3}=n_{3}^{*}$. }\label{fig:eini}
\end{figure}
Figure \ref{fig:eini} illustrates the boosted spatial legs in $e_{a}$-adapted
coordinates, where the components of $n_{1}$ and $\iu n_{2}$ are expressed in
terms of $e_{1}$ and $\iu e_{2}$. While these components are complex, the
tetrads themselves are not complex vectors: each individually represents a
real classical reference frame, and coordinates may always be adapted so that
$n_{a}^{\mu}=\delta^{\mu}_{a}$, or so that $e_{a}^{\mu}=\delta^{\mu}_{a}$,
though not both. The complex quantities reside in the transformations
relating tetrads, whose generators are endowed with the $\ca{3,0}$ Clifford product and
$\mathfrak{su}(2)$ Lie algebra.

\subsection{Tetrad spin structures}\label{sec:spin-ops}
The Clifford product \eqref{eq:CP-Pauli} ensures $\Omega_{i}^{2}=1$ so that the spectrum of each tetrad generator is $\spec(\Omega_{i})= \pm1$ with each eigenvalue being doubly degenerate. The  eigenvectors of $\Omega_{3}$ follow from the Clifford product \eqref{eq:CP-Pauli} and the Lie algebra \eqref{eq:Lie-su2},
which ensure the operators $\K_{\pm}=\mathbf{1}\pm\Omega_{3}$ and $\Omega_{\pm}=\Omega_{1}\pm \iu \Omega_{2}$ satisfy 
\begin{equation}\label{eq:proj-ids}
\Omega_{3}\K_{\pm}=\pm \K_{\pm} \qquad \Omega_{3}\Omega_{\pm} =\pm\Omega_{\pm},
\end{equation}so that the operators $\K_{\pm}$ and $\Omega_{\pm}$  project any timelike direction $e_{0}\in\mathcal{H}$ onto $\Omega_{3}$ eigenvectors. The eigenvectors of $\Omega_{3}$ issuing from the timelike direction $e_{0}$ are arranged in doublets of the form
\begin{equation}\label{eq:null-tetrad-e}
\begin{cases}
\kappa_{1}=\frac{1}{\sqrt{2}}(e_{0}+ e_{3})\\
\kappa_{2}=\frac{1}{\sqrt{2}}\left(e_{1} - \iu e_{2}\right)
\end{cases}
\qquad
\begin{cases}
\kappa_{3}=\frac{1}{\sqrt{2}}\left(e_{1} + \iu e_{2}\right)\\
\kappa_{4}=\frac{1}{\sqrt{2}}(e_{0} - e_{3}).
\end{cases}
\end{equation}
Geometrically $\K_{\pm}$ generate operational null rays propagating along the $e_{3}$-direction while $\Omega_{\pm}$ generate null rays $e_{1}\pm \iu e_{2}$ which span the plane transverse to these operational null rays.    
In coordinates adapted to the tetrad $e_{a}$, where $e_{a}^{\mu}=\delta^{\mu}_{a}$, the tetrad is also orthonormal with respect to the Hermitian inner-product with $\ip{e_{a}}{e_{b}}=e_{a}^{\dag}e_{b}=\delta_{ab}$.  This in turn extends to the null ray tetrad with $\ip{\kappa_{a}}{\kappa_{b}} = \delta_{ab}$.

The doublets $\{\kappa_1,\kappa_2\}$ and $\{\kappa_3,\kappa_4\}$ are spin doublets, with the ladder operators $\Omega_\pm=\Omega_{1} \pm \iu \Omega_{2}$ exchanging members within each doublet. As previously noted, the spatial legs $e_i=\Omega_i e_0$ may be related to a null ray structure $k_i$ according to $e_i=k_i-e_0,$ so that the tetrad $e_a$ may be defined in terms of an observer four velocity $e_0$ and null rays $k_i$. The tetrad $e_a$ may be also recovered from the null ray tetrad $\kappa_a$ by a unitary transformation $e_a=U^{\dag}\kappa_a$ where an explicit representation of $U$ is given in \ref{apx:reps}. As such the null ray tetrad provides an alternative representation of the tetrad $e_a$
\begin{equation}\label{eq:kappa2e}
e_{0} = \frac{1}{\sqrt{2}}(\kappa_1+\kappa_4) \quad e_1=\frac{1}{\sqrt{2}}(\kappa_2 + \kappa_3)\quad e_2=\frac{i}{\sqrt{2}}(\kappa_2-\kappa_3)\quad e_3=\frac{1}{\sqrt{2}}(\kappa_1-\kappa_4)
\end{equation}
and is analogous to the Newman--Penrose tetrad \cite{PR1984}. A four vector $u=u^{a}e_{a}$ can be written in the $\kappa_{a}$ basis as
$u=U^{a}\kappa_{a}$ where the components $U^{a}$ are the individual entries of the matrix
\[u^{a}\sigma_{a}=\frac{1}{\sqrt{2}}\begin{pmatrix}u^0+u^3 & u^1-\iu u^2 \\ u^1 + \iu u^2 & u^0 - u^3 \end{pmatrix}.\] 
Letting $\{\zeta^{r}\}_{r=1}^{2}$ be eigenvectors of $\sigma^{3}$, then the four vector
$u=U^{a}\kappa_{a}$ may be written as $u=\zeta^{r\dag}(u^{a}\sigma_{a})\zeta^{s}k_{rs}$ where
$k_{11}=\kappa_{1}$, $k_{21}=\kappa_{2}$, $k_{12}=\kappa_{3}$ and $k_{22}=\kappa_{4}$. Thus a
four vector may be identified with a spinor bilinear. Moreover, when $u\in\mathcal{H}$ then
$\det(u^{a}\sigma_{a})=\tfrac{1}{2}$ and
$(u^{a}\sigma_{a})^{-1}=2\left(u^{0}\sigma_{0}-u^{i}\sigma_{i}\right)$.

\subsection{Boosts of spin structures}\label{sec:spin-boosts}

The frame dependence of relativistic spin measurements is the central subtlety of relativistic quantum information: spin and momentum cannot be cleanly decoupled under Lorentz boosts, and the reduced spin density matrix is not Lorentz covariant \cite{Cza1997, PST2002, PT2004}. In the tetrad formalism this frame dependence separates into two sharply distinct facts, each following from a short Clifford algebra computation: the spectra of the spin operators are Lorentz invariant, while their expectation values are not, and the mechanism of the latter is the failure of the Hermitian inner product, not of the spin operators themselves,  to be invariant under boosts. For composed boosts in non-colinear directions the same computation yields the Wigner rotation directly.

The projection identities \eqref{eq:proj-ids} make no reference to the timelike direction on which the tetrad generators act, so any $n_{0}\in\mathcal{H}$ has an associated null ray tetrad
\begin{equation}\label{eq:null-tetrad-n}
\lambda_{1} = \frac{1}{\sqrt{2}}(n_0+n_3) \quad  \lambda_{2}=\frac{1}{\sqrt{2}}(n_1-\iu n_2)
\qquad
\lambda_{3}=\frac{1}{\sqrt{2}}(n_1 + \iu n_2) \quad \lambda_{4} = \frac{1}{\sqrt{2}}(n_0-n_{3}),
\end{equation}
cf. equation \eqref{eq:null-tetrad-e}. The null ray tetrad $\lambda_{a}$ is orthonormal with respect to the Hermitian inner product in coordinates adapted to its associated tetrad, where $n_{a}^{\mu}=\delta^{\mu}_{a}$ and $\ip{\lambda_{a}}{\lambda_{b}}=\delta_{ab}$. 
The relation between the null ray tetrads $\lambda_{a}$ and $\kappa_{a}$
follows entirely from the Clifford product \eqref{eq:CP-Pauli} and Lie
algebra \eqref{eq:Lie-su2}, which ensure
\begin{equation}\label{eq:proj-boost}
\begin{aligned}
\K_+\Lambda(b) &= \mathfrak{a}_{+}\K_+ + \bar{\mathfrak{b}}\,\Omega_{+}
&\qquad
\Omega_{-}\Lambda(b) &= \mathfrak{a}_{+}\Omega_{-} + \bar{\mathfrak{b}}\K_-\\
\Omega_{+}\Lambda(b) &= \mathfrak{a}_{-}\Omega_{+} + \mathfrak{b}\K_+
&\qquad
\K_-\Lambda(b) &= \mathfrak{a}_{-}\K_- + \mathfrak{b}\,\Omega_{-}
\end{aligned}
\end{equation}
where
\begin{equation}
 \mathfrak{a}_{\pm}=\cosh(b) \pm \hat{u}^{3}\sinh(b) \qquad \mathfrak{b}=\left(\hat{u}^{1} + \iu \hat{u}^{2}\right)\sinh(b).
\end{equation}
Applying the identities \eqref{eq:proj-boost} to $e_{0}/\sqrt{2}$ yields the
transformation between the null ray tetrads directly:
\begin{equation}\label{eq:block}
\begin{pmatrix}\lambda_{1} \\ \lambda_{2} \\ \lambda_{3} \\ \lambda_{4}\end{pmatrix} = 
\begin{pmatrix}\mathfrak{a}_{+} & 0 & \bar{\mathfrak{b}} & 0 \\ 0 & \mathfrak{a}_{+} & 0 & \bar{\mathfrak{b}} \\ \mathfrak{b} & 0 & \mathfrak{a}_{-} & 0 \\ 0 & \mathfrak{b} & 0 & \mathfrak{a}_{-}\end{pmatrix}\begin{pmatrix}\kappa_{1} \\ \kappa_{2} \\ \kappa_{3} \\ \kappa_{4}\end{pmatrix},
\end{equation}
implemented by a Hermitian but non-unitary matrix
$\mathfrak{M}=\mathfrak{M}^{\dag}$, so that in $e_{a}$ adapted coordinates
the Hermitian norms take the form
\begin{equation}\label{eq:ip-transform}
\begin{cases}
\lambda_{a}^{\dag}\lambda_{a} = \mathfrak{a}_{+}^{2} + \abs{\mathfrak{b}}^{2},\quad a=1,2\\
\lambda_{a}^{\dag}\lambda_{a} = \mathfrak{a}_{-}^{2} + \abs{\mathfrak{b}}^{2},\quad a=3,4.
\end{cases}
\end{equation}

These features have precedents in both the classical and the quantum literature. The action of a boost on null directions, with the Doppler factors appearing as eigenvalues on the light cone and the transverse mixing as aberration, underlies the elementary derivation of the Wigner rotation given in \cite{FT1999}, while the frame dependence of spin expectation values and the coherence induced between spin states are established features of relativistic quantum information \cite{Cza1997, PST2002, PT2004, GA2002, FBHH2010}. What the tetrad formalism contributes is the mechanism: both follow from the single block structure \eqref{eq:block}, in which the Doppler factors $\mathfrak{a}_{\pm}$ are the Hermitian norms of the null rays and the aberration $\mathfrak{b}$ is the doublet mixing.

Noting that $\hat{u}^{3}\in[-1,1]$ and $\mathbf{\hat{u}\cdot\hat{u}}=1$, then $\hat{u}^{3}=\tanh(\beta)$ and $\varphi=\arg\{\hat{u}^{1}+\iu\hat{u}^{2}\}$ yields the transformation
\[\lambda_{1}=\frac{\cosh(b+\beta)}{\cosh(\beta)}\kappa_{1} + \frac{\sinh(b)}{\cosh(\beta)}\mathrm{e}^{-\iu\varphi}\kappa_{3} \qquad \lambda_{3}=\frac{\sinh(b)}{\cosh(\beta)}\mathrm{e}^{\iu\varphi}\kappa_{1} + \frac{\cosh(b-\beta)}{\cosh(\beta)}\kappa_{3}. \]
Figure~\ref{fig:doublet_boost} illustrates the null rays $\{\kappa_{1}$, $\kappa_{3}\}$, and their boosted counterparts $\{\lambda_{1}$, $\lambda_{3}\}$ with reference to $e_{a}$.  The asymmetric stretching of $\kappa_{1}$ and $\kappa_{3}$  and the conjugate admixture phases $\mathrm{e}^{\mp\iu\varphi}$ are directly visible, the latter as the mutually orthogonal imaginary parts $\Im\,\lambda_{1}$ and $\Im\,\lambda_{3}$, while the boost leaves the doublet entirely within the $+1$-eigenspace spanned by $\kappa_{1}$ and $\kappa_{3}$.
\begin{figure}[ht!]
\centering
\includegraphics[width=0.65\textwidth]{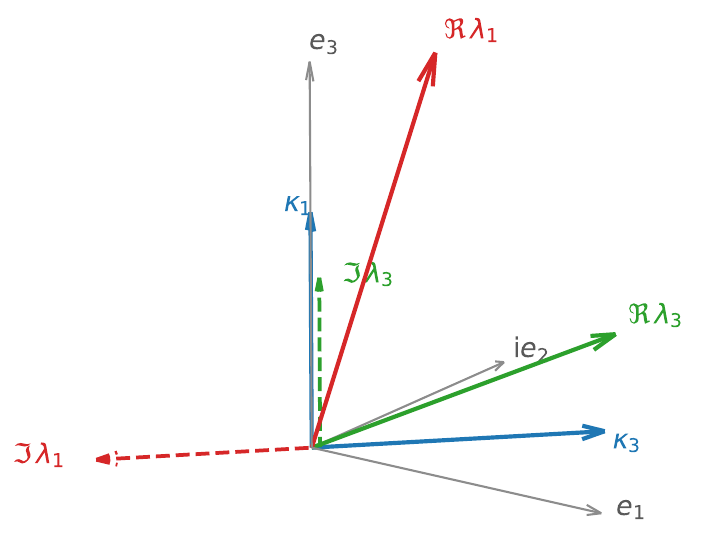}
\caption{The $+1$-eigenspace doublet of $\Omega_{3}$ under the boost
$\Lambda(b)$ in $e_{a}$-adapted coordinates, in the basis
$(e_{1},\iu e_{2}, e_{3})$, with $b=0.8$,
$\hat{u}^{3}=\tanh(\beta)=\cos(70^{\circ})$ and $\varphi=60^{\circ}$. The
conjugate phases of the admixtures appear as the imaginary parts
$\Im\,\lambda_{1}$, transverse, and $\Im\,\lambda_{3}$, longitudinal
($\Im\,\lambda_{3}$ is drawn marginally off-axis for visibility). The
$-1$-eigenspace doublet $\{\kappa_{2}, \kappa_{4}\}$ transforms
identically.}\label{fig:doublet_boost}
\end{figure}

\subsection{Wigner rotations}\label{sec:Wigner}
The null ray tetrads $\kappa_{a}$ and $\lambda_{a}$ are related by $\K_{\pm}$ and $\Omega_{\pm}$ applied to timelike directions $e_{0}$ and $n_{0}=\Lambda e_{0}$. Notably, the tetrads are not related by a naive Lorentz boost of the form $\lambda_{a}=\Lambda\kappa_{a}$ since $\Omega_{3}\kappa_{a}=\epsilon\kappa_{a}$ ($\epsilon=\pm 1$) does not imply $\Omega_{3}\Lambda\kappa_a=\epsilon\Lambda\kappa_{a}$. This is because $\comm{\Omega_{3}}{\Lambda}\neq0$ unless the boost takes place in $(e_0,e_{3})$-plane only.

The tetrad generators $\Omega_{i}$ used here have chiral counterparts $\Omega_{i}^{*}$ arising as the complementary solution of equation \eqref{eq:pseudo} whose canonical representation is the complex conjugate of $\Omega_i$ from \ref{apx:reps}. These yield the pure boost generators $K_{i}$ and pure rotation generators $J_{i}$ according to
\[K_{i}=\frac{1}{2}(\Omega_{i}+\Omega_{i}^{*})\quad J_{i}=\frac{1}{2\iu}(\Omega_{i}-\Omega_{i}^{*}),\]
see \cite{Car2000} for example. 
This now means that conjugation of $\Omega_{i}$ by a pure boost with rapidity $b$ and direction $\mathbf{\hat{u}}\in S^{2}$ is implemented by
\begin{equation}
\Omega_{i}(b)=\mathrm{e}^{b\hat{u}^{j}(\Omega_{j}+\Omega_{j}^{*})/2} \Omega_{i}\mathrm{e}^{-b\hat{u}^{j}(\Omega_{j}+\Omega_{j}^{*})/2}= \mathrm{e}^{b\opr{U}/2} \Omega_{i} \mathrm{e}^{-b\opr{U}/2}. 
\end{equation} 
Likewise, a pure rotation is implemented by
\begin{equation}
\tilde{\Omega}_{i}(\theta)= \mathrm{e}^{-\iu \theta \opr{N}/2} \Omega_{i} \mathrm{e}^{\iu \theta \opr{N}/2},\quad \opr{N}=\hat{n}^{j}\Omega_{j}
\end{equation}  
see \cite{ZR1992} for further discussion.

We now consider a third tetrad $w_{a}$ moving in relation to the tetrads $e_{a}$ and $n_{a}$.  Under a pure boost 
\[B(a)w_0=e^{a(\opr{U} + \opr{U}^*)/2}w_0 =e_0\]
the spatial legs  $w_{i}$ are transported to $\tilde{w}_{i}^{e} $ according to 
\[B(a)\Omega_{i}w_0 = \Lambda(a/2)\Lambda(a/2)^*\Omega_{i}w_0 = \tilde{w}_{i}^{e}.\] 
Similarly under $n_{0}=\Lambda(b)w_0$ the spatial directions $w_{i}$ are transported to $\tilde{w}^{n}_{i}$ as
\[B(b)\Omega_{i}w_0 = \Lambda(b/2)\Lambda(b/2)^*\Omega_{i}w_0=\tilde{w}_{i}^{n} \qquad B(b)w_0=\mathrm{e}^{b(\opr{V} + \opr{V}^*)/2}w_0.\]
The relationship between $\tilde{w}_{i}^{e}$ and $\tilde{w}_{i}^{n}$ is now given by
\[\tilde{w}_{i}^{e} = \Lambda(a/2)\Lambda(-b/2)\Lambda(a/2)^{*}\Lambda(-b/2)^{*}\tilde{w}_{i}^{n}.\]
Expansion of  $\Lambda(a/2)\Lambda(-b/2)$ now gives
\begin{equation}\label{eq:composite}
\Lambda(a/2)\Lambda(-b/2) = \left[\cosh(a/2)\mathbf{1}+\sinh(a/2)\opr{U}\right]\left[\cosh(b/2)\mathbf{1}-\sinh(b/2)\opr{V}\right]
\end{equation}
and in particular the bi-vector part of the term
\[\opr{U}\opr{V}=\hat{u}^{i}\hat{v}^{j}\delta_{ij}\mathbf{1} + \iu\varepsilon_{ijk}\hat{u}^{i}\hat{v}^{j}\Omega_{k}\]
induces a rotation in the plane spanned by $\mathbf{\hat{v}}$ and $\mathbf{\hat{u}}$.

Let us now reconsider the results of Terashima and Ueda \cite{TU2003} from the perspective of tetrad generators, where the tetrad $e_{a}$ moves in the $w_{1}$-direction ($\hat{u}^{i}= \delta^{i}_{1}$) and $n_{a}$ moves in the $w_{3}$ direction ($\hat{v}^{i}= \delta^{i}_{3}$). The composite tetrad rotor $L=\Lambda(a/2)\Lambda(-b/2)$ is of the form
\begin{equation}
L=l_{0}\mathbf{1} + l_{1}\Omega_{1} - l_{3}\Omega_{3} + l_{31}\Omega_{3}\Omega_{1}.
\end{equation}
The bi-vector term $l_{31}\Omega_{3}\Omega_{1}$ indicates a rotation in the $w_{1}-w_{3}$ plane i.e.\ about the $w_{2}$ direction. This role of $\Omega_{3}\Omega_{1}$ as a rotation generator is particularly clear in the tetrad formalism since 
$$\Omega_{3}\Omega_{1}w_{1}=\Omega_{3}\Omega_{1}\Omega_{1}w_{0}=\Omega_{3}w_0=w_{3}$$ 
so that $\Omega_{1}$ first maps $w_{1}$ to $w_0$ then $\Omega_{3}$ maps $w_{0}$ to $w_{3}$. A similar process occurs when applied to $w_{3}=\Omega_{3}w_0$  upon using $\Omega_{3}\Omega_{1}=-\Omega_{1}\Omega_{3}$ to give $\Omega_{3}\Omega_{1}w_{3}=-w_{1}$. Meanwhile the direction $w_{2}$ is unchanged since $\Omega_{3}\Omega_{1}=\iu\Omega_{2}$ which commutes with $\Omega_{2}$. Thus $\Omega_{3}\Omega_{1}$ rotates $w_{1}$ towards $w_{3}$  about $w_{2}$.  

This now indicates $L$ is a composite rotor of a boost and rotation $L=BR$  of the form 
\begin{equation}
L=(b_{0}\mathbf{1} + b_{1}\Omega_{1} + b_{3}\Omega_{3})(r_{0}\mathbf{1}+ r_{31}\Omega_{3}\Omega_{1})  = l_{0}\mathbf{1} + l_{1}\Omega_{1} - l_{3}\Omega_{3} + l_{31}\Omega_{3}\Omega_{1}.
\end{equation}
Since 
\[b_{0}\mathbf{1} + b_{1}\Omega_{1} + b_{3}\Omega_{3} = (l_{0}\mathbf{1} + l_{1}\Omega_{1} - l_{3}\Omega_{3} + l_{31}\Omega_{3}\Omega_{1})(r_{0}\mathbf{1} - r_{31}\Omega_{3}\Omega_{1}) \]
is a pure boost, rotations generated by $\Omega_{3}\Omega_{1}$ must be absent from the right hand expression, which now requires $r_{31}/r_{0}=l_{31}/l_{0}$. It is clear from equation \eqref{eq:composite} that $l_{0}=\cosh(a/2)\cosh(b/2)$ and $l_{31}=\sinh(a/2)\sinh(b/2)$ while the coefficients of $R$ are of the form $r_0=\cos(\delta/2)$ and $r_{31}=\sin(\delta/2)$. It now follows that
\begin{equation}
\tan(\delta)=\frac{\sinh(a)\sinh(b)}{\cosh(a)+\cosh(b)}, 
\end{equation}
so $\delta$ is the standard Wigner rotation angle (cf. reference \cite{TU2003}). The key advantage of the tetrad approach is that the decomposition $L = BR$ is read off directly from the Clifford algebra structure of the composite rotor, without requiring the explicit matrix calculations used in \cite{FT1999, TU2003,  Hal1968}.

\subsection{Relativistic spin correlations}\label{sec:singlet}
The isotropic singlet $\Psi$ is defined by the condition $(\Omega_{i}\otimes 1 + 1\otimes \Omega_{i})\Psi=0$ and may be written with respect to the $\Omega_{3}$ eigenbasis as:
\begin{equation}
\Psi=\frac{1}{\sqrt{2}}\kappa_{1}\otimes\kappa_{2} - \frac{1}{\sqrt{2}}\kappa_{2}\otimes\kappa_{1}.\
\end{equation}  
from which it is found
\begin{equation}\label{eq:isotropic_spin}
(\Omega_{i}\otimes\mathbf{1} + \mathbf{1}\otimes\Omega_{i})\Psi = 0,\quad \forall \ i\in \{1,2,3\},
\end{equation}
see \cite{Bel2004} for example. 
It is clear that
\[\Omega_{i}\otimes\mathbf{1}\Psi =-\mathbf{1}\otimes \Omega_{i}\Psi\]
so that upon applying $1\otimes \Omega_{j}$ the Lie bracket \eqref{eq:Lie-su2} now ensures 
\[\Omega_{i}\otimes \Omega_{j}\Psi = -\mathbf{1}\otimes \left(\delta_{ji}+2\iu \varepsilon_{jik}\Omega_{k}\right)\Psi.\]
This now ensures the correlation tensor of the isotropic singlet state is 
\begin{equation}\label{eq:correlation}
\ev{\Omega_{i}\otimes \Omega_{j}}{\Psi}=-\delta_{ij}
\end{equation} 
as expected. Let us now consider the meaning of the isotropic spin state $\Psi$: It is recalled that 
\[\kappa_{1}=\frac{1}{\sqrt{2}}\K_{+}e_{0}\qquad \kappa_{2}=\frac{1}{\sqrt{2}}\Omega_{-}e_0.\]
If $e_{a}$ and $n_{a}$ are two co-moving tetrads with $e_{0}=n_{0}$ and associated null ray tetrads $\kappa_{a}=\lambda_{a}$, then $\Psi$ is the entangled state of these null ray tetrads, so that in tetrad formalism we have 
\[\Psi = \frac{1}{2}\K_{+}e_0\otimes\Omega_{-}n_0 - \frac{1}{2}\Omega_{-}e_0\otimes\K_{+}n_0,\qquad e_0=n_0.\]

We now consider the case where $e_{a}$ and $n_{a}$ are in uniform motion so that $n_{0}=\Lambda(b)e_{0}$. The essential point of the tetrad formalism is that $\Omega_{i}$ are universal and have already been used to develop $\K_{\pm}$ and $\Omega_{\pm}$ which generate null ray tetrads $\lambda_{a}$ and $\kappa_{a}$ attached to each moving frame. The $\Omega_{3}$ eigenvalues are the same for corresponding members of $\kappa_{a}$ and $\lambda_{a}$, irrespective of their relative velocity. However the Hermitian norms of each tetrad are frame dependent with $\lambda_{1}$ and $\lambda_{2}$ and their respective Hermitian norms given in equations \eqref{eq:null-tetrad-n} and \eqref{eq:ip-transform}. Denoting their normalised counterparts by
$\hat{\lambda}_{a}=\lambda_{a}/\sqrt{\ip{\lambda_{a}}{\lambda_{a}}}$, the projection identities \eqref{eq:proj-ids} ensure
\[\Omega_{3}\hat{\lambda}_{1}=\hat{\lambda}_{1} \qquad \Omega_{3}\hat{\lambda}_{2}=-\hat{\lambda}_{2} \qquad \Omega_{1}\hat{\lambda}_{1}=\hat{\lambda}_{2},\]
analogous to the relations satisfied by $\kappa_{1}$ and $\kappa_{2}$.  The new antisymmetric state
\begin{equation}
\Psi = \frac{1}{\sqrt{2}}\left(\kappa_{1}\otimes\hat{\lambda}_{2} - \kappa_{2}\otimes\hat{\lambda}_{1}\right)
\end{equation}
also satisfies $(\Omega_{i}\otimes 1 + 1\otimes\Omega_{i})\Psi=0$ for $i=1,2,3$, cf. equation \eqref{eq:isotropic_spin}. It is also straightforward to verify the single particle expectation values
$\ev{1\otimes\Omega_{i}}{\Psi}$ vanish, while the correlation tensor is
\begin{equation}\label{eq:corr-boosted}
\left<\Psi\vert\Omega_{i}\otimes \Omega_{j}\vert\Psi\right>=-\delta_{ij},
\end{equation}
cf. equation \eqref{eq:correlation}. Thus the correlation tensor for the spin structures obtained by application of the universal tetrad generators $\Omega_{i}$  to general timelike directions $e_{0}$ and $n_{0}$  is independent of the boost $\Lambda(b)$ relating them.

This is the two-particle counterpart of the invariance established in \S\ref{sec:spin-boosts}. The block form \eqref{eq:block} preserves the $\Omega_{3}$ eigenvalue of each null ray while altering only its Hermitian norm, and normalisation removes the norm, so the correlations are carried entirely by the eigenvalue structure of the tetrad generators, which is universal. Treatments based on rest frame spin operators averaged over momentum wave packets instead find the spin correlations to be degraded under boosts, cf. \cite{Cza1997, PST2002, PT2004, GA2002, FBHH2010}.  The difference lies in which operator is taken to represent relativistic spin, cf. \cite{TA2021}.

\section{Boost invertibility and the Dirac equation}\label{sec:dirac}

\subsection{Correspondence of boost invertibility and the mass shell condition}\label{sec:mass-shell}
The result of Theorem \ref{thm:Clifford-product} means that the tetrad generators $\Omega_{i}\in\cl_{3,0}(\mathbb{R})$ belong to the Pauli algebra. This algebra is isomorphic to the even part $\cl_{1,3}^{+}(\mathbb{R})$  of the spacetime algebra (see \cite{ZR1992} for example). The spacetime algebra $\cl_{1,3}(\mathbb{R})$ is generated by $\{\gamma_{a}\}_{a=0}^{3}$ which are endowed with the Clifford product
\begin{equation}
\gamma_{a}\gamma_{b} + \gamma_{b}\gamma_{a} = 2\eta_{ab},  \label{eq:CP-Dirac}
\end{equation}
see \cite{Lou2001, DL2003}. The isomorphism $\cl_{3,0}(\mathbb{R})\cong \cl_{1,3}^{+}(\mathbb{R})$ is realised by the canonical embedding
\begin{equation}
\Omega_{i} \cong \alpha_i = \gamma_i\gamma_0,
\end{equation} 
see \cite{Hes2015}.  The matrices $\Omega_{i}$ and $\alpha_{i}$ are not equal, rather this isomorphism ensures any algebraic relation among $\Omega_{i}$ has a counterpart realised by $\alpha_i$. In particular, the Clifford product \eqref{eq:CP-Dirac} ensures
\begin{equation}
\alpha_{i}\alpha_{j}+\alpha_{j}\alpha_{i}=2\delta_{ij} \label{eq:pp-sta}
\end{equation}
cf. equation \eqref{eq:CP-Pauli} so that $\alpha_{i}$ in the even spacetime algebra $\cl_{1,3}^{+}(\mathbb{R})$ also furnish a representation of the Pauli algebra $\cl_{3,0}(\mathbb{R})$.

A tetrad map $\Lambda=u^{0}+u^{i}\Omega_{i}$ with $u\in\mathcal{H}$,  has inverse $\Lambda^{-1}=u^{0}-u^{i}\Omega_{i}$ which follows from $\Lambda^{\top}\eta\Lambda=\eta$ and $\eta\Omega_{i}^{\top}\eta=-\Omega_{i}$ (see Lemma \ref{lem:metric-anti-symmtry}). Hence, for any $u\in\mathcal{H}$ it is found that
\begin{equation}
(u^{0}\mathbf{1} + u^{i}\Omega_{i}) (u^0\mathbf{1} - u^{j}\Omega_{j}) = \mathbf{1}, \label{eq:invert-pa}
\end{equation}
in the Pauli algebra. 
The Clifford product \eqref{eq:pp-sta} now means that  the invertibility condition \eqref{eq:invert-pa} in the Pauli algebra, \emph{must}  have a corresponding  form in the spacetime algebra, given by 
\begin{equation}
(u^{0}\mathbf{1} + u^{i}\alpha_{i}) (u^{0} \mathbf{1}- u^{j}\alpha_{j}) = \mathbf{1}.\label{eq:invert-sta}
\end{equation}
The Clifford product \eqref{eq:CP-Dirac} ensures $\gamma_{0}^{2}=\mathbf{1}$ and the involution $\gamma_{0}\gamma_{i}\gamma_{0}=-\gamma_{i}$, while using $\alpha_{i}=\gamma_{i}\gamma_{0}$,  the invertibility condition \eqref{eq:invert-sta} now becomes
\begin{equation}\label{eq:invert-Dirac}
(u^{a}\gamma_{a})^2=\mathbf{1},
\end{equation} 
where $u^{a}\gamma_{a}$ is a Dirac operator. 

A particle of mass $m$  and four momentum $p$ has four velocity $u=p/m\in\mathcal{H}$. The condition \eqref{eq:invert-pa} means that  the tetrad map parameterised by $p/m$ is invertible, and its counterpart in the spacetime algebra is $(p^{a}\gamma_{a})^2=m^2\mathbf{1}$. This is the factorisation of the mass shell condition for which the matrices $\gamma_{a}$ were introduced in \cite{Dir1928}. In the tetrad framework, the Dirac operator $p^{a}\gamma_{a}$ is not an independent postulate introduced to achieve this factorisation. Rather it is the reformulation of tetrad map invertibility when expressed in the spacetime algebra.

\subsection{Dirac spinors}\label{sec:spinors}
The Dirac factorisation $(p^{a}\gamma_{a})^{2}=m^2$  ensures the Dirac operator $p^{a}\gamma_{a}$ has eigenvalues $\pm m$, each with multiplicity two. The corresponding eigenvectors are defined by
\begin{equation}\label{eq:spinor-eqn}
p^{a}\gamma_{a}\omega_{s}^{\pm}(p) = \pm m\omega_{s}^{\pm}(p),
\end{equation} 
where $s$ is an appropriate two-valued index.
In particular, when $p^{0}=m$ and $p^{i}=0$, then the eigenvalue problem becomes
\begin{equation}
\gamma_{0}\xi_{s}^{\pm} = \pm \xi_{s}^{\pm}\quad s=1,2,
\end{equation} 
where $\xi_{s}^{\pm}$ will be called rest spinors. We have found that the map $\phi=\mathrm{e}^{-i\bl{p}{x}/\hbar}$ which generalises the classical action $S=m\tau$, provides the four-momentum $p$ according to $p^{a}=\bar{\phi}\iu\hbar\partial^{a}\phi$. It now follows that equation \eqref{eq:spinor-eqn} is equivalent to 
\begin{equation}\label{eq:Dirac-free}
\iu\hbar\gamma_{a}\partial^{a}\psi = \pm m\psi\qquad \psi=\phi\omega_{s}^{\pm}(p) 
\end{equation}
the free Dirac equation first derived in \cite{Dir1928}. 
 
In the standard treatment, the rest spinors $\xi_{s}^{\pm}$ are defined as simultaneous eigenvectors of $\gamma_{0}$ and $\Sigma_{3}=\diag(\sigma_{3},\sigma_{3})$ (see \cite{Dir2007, Mer1998} for example), while $\Sigma_{3}$ is equivalent to $\Omega_{3}$ under a similarity transformation, see \ref{apx:reps}. In the tetrad formalism, the rest spinors may be analogously defined as simultaneous eigenvectors of $\gamma_{0}$ and $\Omega_{3}$. Moreover, since the null ray tetrad $\kappa_{a}$ is a complete eigenbasis of $\Omega_{3}$, the rest spinors $\xi_{s}^{\pm}$ may be formally identified with this tetrad, where it is also found that 
\[
\gamma_{0}\kappa_{1}=\kappa_{1}, \qquad \gamma_{0}\kappa_{2}=\kappa_{2} \qquad 
\gamma_{0}\kappa_{3}=-\kappa_{3}, \qquad \gamma_{0}\kappa_{4}=-\kappa_{4}.
\]
The four null rays therefore yield the four rest spinors 
$\xi_{s}^{\pm}$ directly: the positive-$\gamma_{0}$ doublet is 
$(\xi_{+}^{+}, \xi_{-}^{+}) = (\kappa_{1}, \kappa_{2})$, and the 
negative-$\gamma_{0}$ doublet is $(\xi_{+}^{-}, \xi_{-}^{-}) 
= (\kappa_{3}, \kappa_{4})$. 

The spinor boost $\mathsf{S}(b)$ satisfies the standard intertwining 
property $u^{\mu}\gamma_{\mu}\mathsf{S}(b) = \mathsf{S}(b)\gamma_{0}$, 
see \cite{Wei1995} for example. It follows that
\begin{equation}
p^{a}\gamma_{a}\omega_{s}^{\pm} = m\, u^{\mu}\gamma_{\mu}\mathsf{S}(b)\xi_{s}^{\pm} 
= m\,\mathsf{S}(b)\gamma_{0}\xi_{s}^{\pm} 
= \pm m\,\mathsf{S}(b)\xi_{s}^{\pm} = \pm m\,\omega_{s}^{\pm},
\end{equation}
so the boosted  spinor $\omega_{s}^{\pm}(b)=\mathsf{S}(b)\xi_{s}^{\pm}$ 
is an eigenspinor of the Dirac operator $p^{a}\gamma_{a}$ with eigenvalue $\pm m$. The boosted spinors $\omega_{s}^{\pm}(b)$ are attained via 
$$\mathsf{S}(b)=\cosh\left(\frac{b}{2}\right)\mathbf{1}+\sinh\left(\frac{b}{2}\right)\hat{u}^{i}\alpha_{i}\qquad \alpha_{i}=\gamma_{i}\gamma_{0}.$$
The null rays are spacetime vectors $\kappa_{a}$, while $\xi_{s}^{\pm}$ are rest spinors, and both are represented by the same set of column vectors. Geometrically these are different objects, since $\kappa_{a}$ transform under $\Lambda(b)$ in the Pauli algebra, while $\xi_{s}^{\pm}$ transform under $\mathsf{S}(b)$ in the spacetime algebra.

\subsection{Dirac currents}\label{sec:currents}
Given a Dirac spinor $\psi$, the associated Dirac current is
\begin{equation}
j_{a}(x) = \bar{\psi}\gamma_{a}\psi\quad \bar{\psi}=\psi^{\dag}\gamma_0.
\end{equation} 
A momentum eigenmode $\psi(x;p)=\mathrm{e}^{-\iu\bl{p}{x}/\hbar}\omega_{s}^{+}(p)$ has an associated Dirac current 
\[j_{a}(x) = \bar{\psi}(x;p)\gamma_{a}\psi(x;p)=\frac{p_{a}}{m}.\]
Since $p_{a}p^{a}=m^2$, it follows $j_{a}(x)j^{a}(x)=1$ for all $x.$ This now means that the tetrad map
\begin{equation}
\Lambda(x) = j^{a}(x)\Omega_{a} \quad \Omega_{0}=\mathbf{1}
\end{equation}
satisfies $\Lambda(x)\eta\Lambda(x)^{\top}\eta = 1$ for all $x$, so that $\Lambda(x)$ is an invertible tetrad transformation at all $x$.

The Dirac equation \eqref{eq:Dirac-free} being linear has a more general solution space comprising general superpositions of momentum eigenmodes. Let us consider the simplest example of such a solution:
\begin{equation}
\psi(x) = \frac{1}{\sqrt{2}}\left[u(b)\mathrm{e}^{-\iu\bl{p}{x}/\hbar} + u(-b)\mathrm{e}^{-\iu\bl{q}{x}/\hbar}\right]\qquad p=m\mathrm{e}^{b\Omega_{1}}e_0\quad q=m\mathrm{e}^{-b\Omega_{1}}e_0.
\end{equation} 
$u(\pm b)=\mathrm{e}^{\pm b\alpha_{1}}\xi_{1}^{+}$. 
The associated Dirac current with components $j_{a}(x)=\bar{\psi}\gamma_{a}\psi$ is given by
\begin{equation}
j(x) = \left[\cosh(b) + \cos(\theta(x))\right]e_0 + \sinh(b)\sin(\theta(x))e_{2}\quad \theta(x)=\frac{2p^1x^1}{\hbar}. 
\end{equation}
Using the Gordon decomposition (see \cite{Mer1998} for example) this current has convective and spin currents given by
\begin{equation}
\begin{aligned}
j_{\mathrm{conv}}(x) &= \cosh(b)\left[1+ \cosh(b)\cos(\theta(x))\right]e_{0}\\
j_{\mathrm{spin}}(x) &= \sinh(b)\left[-\sinh(b)\cos(\theta(x))e_{0} +\sin(\theta(x))e_{2}\right] 
\end{aligned}
\end{equation}
Thus the Dirac current has a convective contribution in the $e_{0}$ direction and a spin current contribution in the $e_{2}$ direction, while the momentum  eigenmodes of $\psi$ are in the $e_{1}$-direction.  

This indicates that these currents do not directly correspond to the trajectory of a particle associated with  $\psi$. Notice that the phases $\mathrm{e}^{-\iu\bl{p}{x}/\hbar} $ and $\mathrm{e}^{-\iu\bl{q}{x}/\hbar} $  serve as velocity potentials whose gradients correspond  to the four velocities of congruences of integral curves. The Dirac current  $j(x)$ is neither tangent nor orthogonal to these integral curves and instead  propagates between these level sets and world lines. If  the Dirac current were to correspond to the actual trajectory of a particle, then this would require $j(x)$ to parameterise an invertible tetrad map $\Lambda(j(x)) = j^{a}(x)\Omega_{a}$. This invertibility then requires
\[\Lambda(j(x))\eta\Lambda(j(x))^{\top}\eta =\mathbf{1} \iff  \bl{ j(x)}{j(x)}=1,\]
for all $x$, which is clearly not the case.  

\subsection{Vector potential interactions}\label{sec:min-coupling}
The preceding subsections have taken the classical trajectory as given and recovered the Dirac operator from the invertibility of the tetrad map it parameterises. In the free case a momentum eigenspinor $\psi(x;p)=\mathrm{e}^{-\iu\bl{p}{x}/\hbar}\omega^{+}_{s}(p)$ has Dirac current $j^{a}=p^{a}/m\in\mathcal{H}$, cf. \S\ref{sec:currents}, and the same $p^{a}$ is returned by the phase gradient, $\iu\hbar\partial^{a}\psi=p^{a}\psi$. The bilinear and  phase gradient agree, and the corresponding tetrad map $j^{a}\Omega_{a}$ is invertible.
However, this only applies  to a momentum eigenmode. A
superposition of the form 
\begin{equation}
\psi=\sum_{r}a_{r}\,\mathrm{e}^{-\iu\bl{p_{r}}{x}/\hbar}\omega^{+}_{s_{r}}(p_{r})
\end{equation}
is not associated with a single phase, since the eigenmodes are parameterised by distinct momenta $p_{r}$. Moreover the spinors $\omega^{+}_{s_{r}}$ are generally distinct. As such there is no classical $S$ from which $\psi$ may be reconstructed. 

Consider now a spinor field in the presence of a potential $A_{a}(x)$, and restrict attention to a single orbit branch of the form $\psi=\mathrm{e}^{-\iu S/\hbar}\omega^{\pm}_{s}(x)$,  cf. \cite{MS1964}. If the associated tetrad map $j^{a}\Omega_{a}$ is to remain invertible
the current must lie on the unit hyperboloid, so that 
\begin{equation}\label{eq:current-kinetic}
mj^{a}=m\bar{\psi}\gamma^{a}\psi = \partial^{a}S-qA^{a},
\end{equation}
the kinetic momentum of \S\ref{sec:potentials}. Along a single orbit it is in principle possible to introduce a path dependent $\chi$ such that $\partial_{a}\chi=A_{a}$. However this only applies along a single integral curve and extending beyond it when $F_{\mu\nu}\neq0$ means that $\chi$ is no longer single valued.

The potential therefore cannot be absorbed in the phase of some $\psi$ which is also well-defined throughout spacetime. Instead the potential must be absorbed into an operator acting on $\psi$, so that the spinor field is governed by 
\begin{equation}\label{eq:Dirac-min}
\gamma_{a}D^{a}\psi = m\psi \qquad D^{a}=\iu\hbar\partial^{a}-qA^{a}.
\end{equation}
The Gordon decomposition of the associated current  now yields 
\begin{equation}\label{eq:gordon-min}
2m\bar{\psi}\gamma^{a}\psi = \left[\bar{\psi}D^{a}\psi +
\overline{(D^{a}\psi)}\psi\right] + \hbar\,\partial_{b}\left(\bar{\psi}\sigma^{ab}\psi\right),
\qquad \sigma^{ab}=\frac{\iu}{2}\comm{\gamma^{a}}{\gamma^{b}},
\end{equation}
cf. \S\ref{sec:currents}. Eikonal solutions of \eqref{eq:Dirac-min} are of the form $\psi=\mathrm{e}^{-\iu S/\hbar}\omega^{\pm}_{s}(x)$ with $\omega$ a general normalised spinor satisfying $\bar{\omega}\omega=1$. The associated Dirac current is given by the  Gordon decomposition 
\begin{equation}\label{eq:current-corrections}
m j^{a} = \left(\partial^{a}S-qA^{a}\right)
+ \iu\hbar\,\bar{\omega}\partial^{a}\omega
+ \frac{\hbar}{2}\partial_{b}\left(\bar{\omega}\sigma^{ab}\omega\right).
\end{equation}
The kinetic momentum of \S\ref{sec:potentials} is the leading term, and the corrections are of order $\hbar$ and arise from the variation of the spinor $\omega$: a connection term, and the divergence of the spin bivector $\bar{\omega}\sigma^{ab}\omega$ of \S\ref{sec:currents}. The identification of current and momentum  holds exactly when the spinor is covariantly constant along the congruence, and the classical tetrad map is recovered as the leading order contribution. Terms $\mathcal{O}(\hbar)$ are a  reflection of the spinor's own variation.

\subsection{Tetrad maps and spin in Coulomb potentials}\label{sec:coulomb}
A charge in the field of a nucleus of charge $Z$ has potential energy  $-\zeta/r$, with $\zeta = Zq^{2}/4\pi = Z\alpha_{\mathrm{el}}\hbar$, and the resulting relativistic Kepler problem is separable, with principal function 
\begin{equation}\label{eq:coulomb-HJ}
S(t,r,\varphi;E,\ell) = Et + \ell\varphi + \int^{r}\sqrt{(E^{2}-m^{2}) +
\frac{2E\zeta}{r'} - \frac{\ell^{2}-\zeta^{2}}{r'^{2}}}\;\dd{r'},
\end{equation}
see \cite{LL2013} for example. The associated four momentum $p^{a}=\partial^{a}S-qA^{a}$ in the plane of orbit has components
\[p^{0}=E+\frac{\zeta}{r}\qquad p_{\hat{r}}=\pm\sqrt{\left(E+\frac{\zeta}{r}\right)^{2}-m^2-\frac{\ell^2}{r^2}}\qquad p_{\hat{\varphi}}=\frac{\ell}{r}.\]
The single orbit spinor $\psi=\mathrm{e}^{-\iu S/\hbar}\omega(x)$ is constructed from $\kappa_1$ with 
\[\omega(x)=\frac{\left(m+p^{0}\right)\mathbf{1}+p^{i}\alpha_{i}}{\sqrt{2m\left(m+p^{0}\right)}}\,\kappa_{1},\]
where $\omega(x)$ satisfies $\gamma_{a}(\partial^{a}S -qA^{a})\omega(x)=m\omega(x)$ and $\bar{\omega}(x)\omega(x)=1$ pointwise. The eikonal  $\psi=\mathrm{e}^{-\iu S/\hbar}\omega(x)$ is a solution of the minimally coupled Dirac equation \eqref{eq:Dirac-min} only to terms $\mathcal{O}(\hbar)$.  

The associated Dirac current is $j^{a}(x)=\bar{\omega}(x)\gamma^a\omega(x) = p^{a}(x)/m$ identically and satisfies $\bl{j(x)}{j(x)}=1$ everywhere. Thus $j^{a}(x)\Omega_{a}\in\mathrm{SO}(1,3)$ is a classical invertible coordinate transformation. The map between the lab frame $e_{a}$ and the moving frame $n_{a}$ is parametrised by $j(x)$, and it is this current which varies along the orbital path as opposed to the tetrads $e_{a}$ or $n_{a}$ themselves.  In this sense the acceleration is induced by the tetrad maps between the frames rather than a feature of either frame separately, and neither is privileged during their motion.

Let us denote $\omega(x)=\mathsf{S}(x)\kappa_{1}$  and define the intrinsic spin operator by $\Omega_{3}(x)=\mathsf{S}(x)\Omega_{3}\mathsf{S}^{-1}(x)$. Then the isometry $\mathsf{S}^{\dag}(x)\gamma_{0}\mathsf{S}(x)=\gamma_{0}$ ensures
\begin{equation}
\bar{\omega}(x)\,\Omega_{3}(x)\,\omega(x) = \bar{\kappa}_{1}\Omega_{3}\kappa_{1}
= 1
\end{equation}
at every point of the orbit and for every rapidity. The lab frame spin operator $\Omega_3$ has expectation value
\begin{equation}
\bar{\omega}(x)\,\Omega_{3}\,\omega(x) = \frac{p^{0}(x)}{m}
= \frac{1}{m}\left(E+\frac{\zeta}{r}\right),
\end{equation}
when the boost direction is orthogonal to $e_3$. This oscillates between perihelion and aphelion and returns to its initial value after each radial period. 

\subsection*{Summary}
The identification of the Lorentz transformation between rest and observer frames with a spinor obeying an equation of Dirac form is established by G\"ursey \cite{Gur1957} and Baylis \cite{Bay1992}, where the Pauli algebra is postulated and the moving frame is of Frenet--Serret type. Here neither is assumed: the Clifford product is Theorem \ref{thm:Clifford-product}, while $e_{i}=\Omega_{i}e_{0}$ is generated from the four velocity alone and so is defined for every $e_{0}\in\mathcal{H}$, including the inertial case in which a Frenet frame degenerates. It is well established that spinors may be generated by projectors  $\tfrac{1}{2}(\mathbf{1}+\Omega_{3})\in \cl_{3,0}(\mathbb{R})$, see \cite{Lou2001} for example. Similar  projectors are also used to isolate the columns of a Lorentz transformation in \cite{Bay1992}. In each of these cases  $\Omega_{3}$ is a generator of a postulated algebra. In the current framework $\Omega_{3}$ is a tetrad generator, so that the projection generators $\K_{\pm}$ and the resulting null ray tetrad $\kappa_{a}$ are built from the same operators which produce the spatial legs $e_{i}=\Omega_{i}e_{0}$. The eikonal treatment of the Dirac equation is due to Pauli  and was completed by Rubinow and Keller \cite{RK1963}, who seek a solution
\begin{equation}\label{eq:RK}
\psi \sim \mathrm{e}^{\iu S/\hbar}\sum_{n=0}^{\infty}(-\iu\hbar)^{n}a_{n},
\end{equation}
satisfying \eqref{eq:Dirac-min}. This ensures $S$ satisfies the relativistic Hamilton--Jacobi equation and each $a_{n}$ satisfies an ordinary differential equation along the classical trajectories. The construction here is the leading term of \eqref{eq:RK}, but the direction of the argument is reversed: the Dirac equation is taken as given in \cite{RK1963} and the Hamilton--Jacobi equation recovered from it, whereas here the potential is carried by $D^{a}$ because the failure of $\chi$ to be single valued when $F_{\mu\nu}\neq0$ prevents it from being carried by the phase.

\section{Spin-momentum coupling}
The null ray tetrads $\kappa_{a}$ and $\lambda_{a}$ play dual roles:
\begin{enumerate}[\bfseries i)]
 \item They are eigenvectors of $\Omega_{3}$ transforming as geometric vectors under Lorentz boosts $\Lambda(b)$ generated by $\Omega_{i}$.
 \item They are eigenvectors of $\Omega_{3}$ and $\gamma_{0}$ serving as rest spinors $\xi_{s}^{\pm}$ transforming under Lorentz boosts $\mathsf{S}(b)$ generated by $\alpha_i = \gamma_i\gamma_0 \cong \Omega_i$.
\end{enumerate}
In Section~\ref{sec:spin-boosts}, the transformation between these tetrads was shown to be block-diagonal. This structure preserves the $\pm1$ eigenspaces of $\Omega_3$ under arbitrary boost compositions, preventing eigenspace mixing during Wigner rotations (Section~\ref{sec:Wigner}).

\subsection{Intrinsic spin operators}\label{sec:intrinsic}
Taillebois and Avelar \cite{TA2021} define an intrinsic physical property via an observer-independent mathematical object $A(v)$ tied to an observer's four-velocity $v\in\mathcal{H}$. This requirement automatically excludes standard spatial three-vectors, as their form is tied to a specific hyperplane.

In our framework, the universality condition \eqref{eq:ortho} ensures that the generators $\Omega_{i}\in\mathfrak{so}(1,3)$ generate an orthogonal frame from any $e_{0}\in\mathcal{H}$ and implement tetrad boosts. Consequently, the $\mathfrak{su}(2)$ structure (Corollary~\ref{cor:su2}) is intrinsic, with bivectors $\frac{\iu}{2}\comm{\gamma_i}{\gamma_j}=\varepsilon_{ijk}\Omega_{k}=\varepsilon_{ijk}(K_{k}+\iu J_{k})$ corresponding to the spin bivector $S^{ab}$ in \cite{TA2021}. Transforming $\gamma_{\mu}$ via $\gamma_{\mu}(b)=\mathsf{S}(b)\gamma_{\mu}\mathsf{S}(-b)$ \cite{ZR1992} yields the intrinsic operators:
\begin{equation}
\Omega_{i}(b)=\mathsf{S}(b)\Omega_{i}\mathsf{S}(-b).
\end{equation}

Quantum mechanically, an observable $A(v)$ must satisfy $\Tr[\rho\,A(v)]=\Tr[\rho'\,A(v+\delta v)]$ under observer frame changes \cite{TA2021}. While usually assumed unitary, this relation holds for non-unitary spinor boosts $\kappa_a\to\mathsf{S}(b)\kappa_a$ because $\mathsf{S}(-b)=\gamma_{0}\mathsf{S}^{\dagger}(b)\gamma_{0}$ \cite{Sch1961}, making the boost an isometry of the indefinite inner product $\bar{\kappa}_{a}\kappa_b=\kappa_{a}^{\dag}\gamma_0\kappa_b$. 

For a positive-energy wave packet $\psi = \int \frac{\dd[3]{p}}{2p^0}\sum_{s=1}^{2}a_{s}(b,\hat{u})e^{-\iu\langle p,x\rangle/\hbar}\mathsf{S}(b)\kappa_{s}$ normalised by $\Tr[\rho\gamma_{0}]=1$, evaluating the trace directly confirms compatibility:
\begin{equation}
\Tr[\rho \gamma_0\Omega_{i}(b)] = \frac{m^2}{2}\iiint\dd[2]{\hat{u}}\,\dd{b} \sum_{s=1}^{2}\sinh^2(b)\abs{a_{s}(b,\hat{u})}^2\,\bar{\kappa}_{s}\Omega_{i}\kappa_{s}.
\end{equation}
Thus, $\gamma_{0}\Omega_{i}(b)$ is a valid intrinsic spin operator under non-unitary transformations of $\rho$.

\subsection{Quantum reference frames and tetrads}\label{sec:qrf-comparison}
Giacomini, Castro-Ruiz, and Brukner \cite{GCB2019b} describe relativistic spin in Quantum Reference Frames (QRF) via a momentum-dependent superposition of Lorentz boosts. The tetrad framework offers a complementary geometric view based on the invertibility of boosts $\Lambda(b)\in\cl_{3,0}(\mathbb{R}) \cong \cl_{1,3}^{+}(\mathbb{R})$ (\S\ref{sec:mass-shell}).

In QRF \cite{GCB2019b}, a system $A$ with internal degree of freedom $\tilde{A}$ and lab frame $C$ are transformed via the unitary operator $\hat{S}_{L}=\mathcal{P}_{CA}^{(v)}U_{\tilde{A}}(\hat{\pi}_{C})$ \cite{GCB2019a}, giving $\ket{\Psi}_{A\tilde{A}}^{(C)} = \int \dd{\mu}_{A}(p_{A})\,\psi\left(-\frac{m_{C}}{m_{A}}p_{A}\right)\ket{p_{A};\Sigma_{p_{A}}}_{A\tilde{A}}$. In contrast, a tetrad momentum eigenstate $\psi(x;p)=e^{-\iu\langle p,x\rangle/\hbar}\mathsf{S}(b)\kappa_{s}$ yields a Dirac current $\bar{\omega}_{s}(p)\gamma^{\mu}\omega_{s}(p) = p^{\mu}/m = u^{\mu}$, which parameterises the classical, invertible boost $n_0=\bar{\omega}_{s}(p)\gamma^{a}\omega_{s}(p)\,\Omega_{a}e_0$. For a general wave packet $\psi(x) = \int \frac{\dd[3]{p}}{2p^0}\sum_{s}a(p,s)\,e^{-\iu\langle p,x\rangle/\hbar}\mathsf{S}(b)\kappa_{s}$, $j^a(x)j_a(x) \neq 1$, so $\psi(x)$ generates a superposition of classical tetrad coordinate maps rather than a single classical transformation.

The two constructions correspond directly. The QRF framework constructs laboratory observables $\hat{\Xi}_{i}=\hat{S}_{L}(\hat{\sigma}_{i}\otimes 1)\hat{S}_{L}^{\dagger}$ using a unitary operator constructed from the momentum operator $\hat{\pi}_C$, In the tetrad framework the operator $\Omega_{i}(b)=\mathsf{S}(b)\Omega_{i}\mathsf{S}(-b)$ uses the Clifford element $\mathsf{S}(b)$ with real parameter $b$. Using the unitary representation transformation $U$ (see \ref{apx:reps}), the positive-$\gamma_0$ rest-spinor boost corresponds directly to the QRF state:
\begin{equation}\label{eq:branch-match}
U^{\dagger}\omega_{s}^{+} = U^{\dagger}\mathsf{S}(b)\xi_{s}^{+} = \opr{U}(L_{p_{A}})\ket{k,\vec{\sigma}} = \ket{p_{A};\Sigma_{p_{A}}},
\end{equation}
with the negative-$\gamma_0$ doublet $(\kappa_3, \kappa_4)$ acting as the antiparticle sector projected out in \cite{GCB2019b}.

The key operational difference lies in where the indefiniteness arises. $\hat{\Xi}_{i}$ acts on $\mathscr{H}_{\tilde{A}}\otimes\mathscr{H}_{C}$ and entangles spin with reference-frame momentum via $\hat{S}_{L}$. $\Omega_{i}(b)$ acts strictly on the spinor space with real parameter $b\in\mathbb{R}$; momentum indefiniteness is carried entirely by the wave packet amplitudes acting on definite classical tetrad boosts $j^{a}\Omega_{a}$. Furthermore, while frame invariance in QRF requires unitarity of $\hat{S}_L$, tetrad invariance relies on the isometry $\mathsf{S}^{\dagger}(b)\gamma_{0}\mathsf{S}(b)=\gamma_{0}$.

This comparison stems from the dual geometric/spinorial role of the null-ray tetrad $\kappa_{a}$. In $\cl_{3,0}(\mathbb{R})$, $\kappa_{a}$ are geometric vectors transforming under $\Lambda(b)=\exp(\hat{u}^{i}\Omega_{i}b)$. Under the embedding $\cl_{3,0}(\mathbb{R})\cong\cl_{1,3}^{+}(\mathbb{R})$, they transform as standalone spinors under $\mathsf{S}(b)=\exp(\hat{u}^{i}\alpha_{i}b/2)$, reconstructing the spacetime tetrad bilinearly through the Dirac current $n_{0}^{\mu}=\bar{\omega}_{s}(p)\gamma^{\mu}\omega_{s}(p)$. The spin operators $\Omega_i(b)$ are thus the spinorial counterparts of the geometric tetrad generators.

\subsection{The hydrogenic ground state}\label{sec:hydrogenic}
Stationary states of hydrogenic atoms are governed by
\begin{equation}\label{eq:hydrogen}
\left(-\iu\hbar\alpha_{i}\partial^{i}-\frac{\zeta}{r} + \gamma_{0} m\right)\psi=E\psi,
\end{equation}
with $\zeta$ as in \S\ref{sec:coulomb}.
Radial spinors $a(r)\kappa_{a}$ satisfy $\alpha_{i}\partial^{i}a(r)\kappa_{a}=\hat{\rho}\,a'(r)\kappa_{a}$ with $\hat{\rho}=\hat{r}^{i}\alpha_{i}$, and since $\partial^{i}\hat{r}^{j}= r^{-1}(\delta^{ij}-\hat{r}^{i}\hat{r}^{j})$ one finds $(\alpha_{i}\partial^{i})^{2}a(r)\kappa_{a} = \left(a''(r)+2a'(r)/r\right)\kappa_{a}$, so $\alpha_{i}\partial^{i}$ leaves the space of radial spinors spanned by $\{\kappa_{a},\hat{\rho}\kappa_{a}\}$ invariant. Since $\gamma_{0}\alpha_{i}=-\alpha_{i}\gamma_{0}$, the pair
$\{\kappa_{a},\hat{\rho}\kappa_{a}\}$ belongs to different $\gamma_{0}$ doublets, so radial solutions in the positive $\Omega_{3}$ eigenspace take the form
\begin{equation}\label{eq:radial-ansatz}
\psi(r) = \left[a(r) + b(r)\hat{\rho}\right]\kappa_{1}.
\end{equation}
Since $\kappa_{1}^{\dag}\hat{\rho}\kappa_{1}=0$, the Dirac equation separates into two systems for $\kappa_{1}$ and $\hat{\rho}\kappa_{1}$, with coefficients
\begin{align}
-\iu\hbar\left(b'(r)+2\frac{b(r)}{r}\right) &= \left(E-m+\frac{\zeta}{r}\right)a(r), \label{eq:Dirac_kappa} \\
-\iu\hbar\, a'(r) &= \left(E + m + \frac{\zeta}{r}\right) b(r),\label{eq:Dirac_rhokappa}
\end{align}
respectively, cf.\ \cite{Dir2007}.

Normalisability of \eqref{eq:Dirac_kappa}--\eqref{eq:Dirac_rhokappa} gives $a(r)=A(r)F(r)$ and $b(r)=A(r)G(r)$, where
\begin{equation}\label{eq:bound-form}
A(r)=r^{\varsigma-1}\mathrm{e}^{-\chi r},\qquad
\varsigma=\sqrt{1-\frac{\zeta^{2}}{\hbar^{2}}}=\sqrt{1-Z^{2}\alpha_{\mathrm{el}}^{2}},\qquad
\hbar\chi=\sqrt{m^{2}-E^{2}},
\end{equation}
with $\chi\in\mathbb{R}$ for bound states and $F$, $G$ terminating series. The ground state corresponds to $G(r)/F(r)=\iu c_{0}$ with $c_{0}$ constant, so
\begin{equation}\label{eq:ground-state}
\psi(r) = A(r)\left[1 + \iu c_{0}\hat{\rho}\right]\kappa_{1},
\end{equation}
and substitution in \eqref{eq:Dirac_kappa}--\eqref{eq:Dirac_rhokappa} fixes
\begin{equation}\label{eq:c0}
c_{0}=\sqrt{\frac{1-\varsigma}{1+\varsigma}}=\frac{1-\varsigma}{Z\alpha_{\mathrm{el}}},
\qquad \text{so that}\qquad \frac{2c_{0}}{1+c_{0}^{2}}=Z\alpha_{\mathrm{el}}.
\end{equation}
The corresponding Dirac current $j^{a} = \bar\psi\gamma^{a}\psi$ has
\begin{equation}\label{eq:j0-hydrogen}
j_{0}(x) =A(r)^{2}\kappa_{1}^{\dag}[1- \iu c_{0}\hat{\rho}][1+\iu c_{0}\hat{\rho}]\kappa_{1} = A(r)^{2}\left(1 + c_{0}^{2}\right),
\end{equation}
since $\hat{\rho}^{2}=\mathbf{1}$, while the same calculation for $j^{i}$, using $\hat{\rho}\alpha_{i}\hat{\rho}=\hat{\rho}\acomm{\alpha_{i}}{\hat{\rho}}-\alpha_{i}$, $\comm{\alpha_{i}}{\alpha_{j}}=2\iu\varepsilon_{ijk}\Omega_{k}$ and $\kappa_{1}^{\dag}\Omega_{i}\kappa_{1}=0$ for $i=1,2$, gives $j^{1}=-2c_{0}A(r)^{2}\hat{r}^{2}$ and $j^{2}=2c_{0}A(r)^{2}\hat{r}^{1}$. The current $j=j^{a}\Omega_{a}e_{0}$ is therefore purely azimuthal,
\begin{equation}\label{eq:current}
  j(x) = A(r)^{2}\big(1+c_{0}^{2}\big)\,e_{0} + 2A(r)^{2} c_{0}\sin(\theta)\, \tilde{\Omega}_{2}(\varphi) e_{0},
  \qquad \tilde{\Omega}_{i}(\varphi) = \mathrm{e}^{-\iu\varphi\Omega_{3}/2}\, \Omega_{i}\,
    \mathrm{e}^{\iu\varphi\Omega_{3}/2}.
\end{equation}
Dividing by $j_{0}$ and using \eqref{eq:c0}, the local three velocity relative to the lab tetrad is $\mathbf{v}(x)= Z\alpha_{\mathrm{el}}\sin(\theta)\,\tilde{\Omega}_{2}(\varphi) e_{0}$. The radial profile cancels identically, so the speed depends on the polar angle alone and is everywhere the classical orbital velocity $Z\alpha_{\mathrm{el}}$ up to the factor $\sin(\theta)$, with $\abs{\mathbf{v}} < 1$ for $Z < 137$.

The hydrogenic ground state also discriminates between candidate relativistic spin operators. Bauke \emph{et al}. \cite{BAKG2014} require such an operator to commute with the free Dirac Hamiltonian, to have eigenvalues $\pm\tfrac{1}{2}$ and to satisfy the angular momentum algebra, and show that the ground state expectation values of the various candidates separate as $Z$ increases, so providing an experimental test. The tetrad generators meet the second and third conditions, while $\Omega_{3}$ is unitarily equivalent to $\Sigma_{3}$ (\ref{apx:reps}) and so behaves as the Pauli operator of \cite{BAKG2014} in this respect. The intrinsic operator $\gamma_{0}\Omega_{i}(b)$ of \S\ref{sec:intrinsic} falls outside that classification, being an isometry of the indefinite inner product rather than Hermitian with inner product, and its behaviour in the state \eqref{eq:ground-state} remains to be determined.

\section{Conclusion and outlook}\label{sec:conclusion}
This paper has shown that a single condition --- that a fixed system of linear operators $\Omega_{i}$ generates an orthogonal tetrad from every four velocity $e_{0}\in\mathcal{H}$ --- is sufficient to endow every tetrad with an intrinsic spin structure. Universality alone forces the $\mathfrak{so}(1,3)$, $\cl_{3,0}(\mathbb{R})$ and $\mathfrak{su}(2)$ structures of the generators, so that the Pauli algebra, the spin $\mathfrak{su}(2)$, and the imaginary unit as pseudoscalar arise as consequences of the tetrad geometry rather than as separate postulates. The same generators serve three roles at once: they build the spatial legs of a tetrad, they generate the boosts between tetrads, and they carry the intrinsic spin structure across them. Within this framework the Dirac operator and its spinors emerge from the invertibility of boosts between tetrads under the embedding $\cl_{3,0}(\mathbb{R})\cong\cl_{1,3}^{+}(\mathbb{R})$, and the Wigner rotation is read directly from the composite rotor without half-angle spinor manipulation.

Several physical consequences follow. The spin correlations of an entangled pair of null ray tetrads are independent of the boost relating them, the block structure of the transformation preserving the $\Omega_{3}$ eigenvalues while altering only the Hermitian norms. Minimal coupling is not imposed but forced: the potential is carried by $D^{a}$ precisely because the path dependence of the phase when $F_{\mu\nu}\neq0$ prevents it from being carried by the phase of a spinor field defined throughout spacetime. For a Coulomb potential the classical Hamilton--Jacobi congruence supplies a single orbit spinor whose Dirac current is exactly the kinetic momentum, so that the tetrad map it parameterises is invertible at every point of the orbit, while the hydrogenic ground state follows from the Clifford product alone and carries a purely azimuthal current whose local speed is the classical orbital velocity $Z\alpha_{\mathrm{el}}$ up to a factor $\sin(\theta)$.

Set against the quantum-reference-frame description of relativistic spin, the tetrad account offers a complementary view of the same physics. The indefinite relative motion between two frames is described in the QRF formalism by an operator-valued superposition of boosts on the Hilbert space of a quantum frame. It appears here as a superposition of definite classical coordinate transformations between tetrads, carried by the amplitudes of a Dirac spinor rather than by any indefiniteness in the frames themselves. The QRF and tetrad descriptions agree on each momentum eigenstate, and the geometric generators $\Omega_{i}$ underlie both the spin observable in the spacetime algebra and the frame transformation in the Pauli algebra.

Two limitations should be stated plainly: First, the tetrad map associated with a Dirac current is a classical one only for a single momentum in eikonal form. A general solution of the Dirac equation has an associated current such that $j^{a}j_{a}\neq1$, and then $j^{a}\Omega_{a}$ is a superposition of tetrad maps rather than a single invertible one. Second: the tetrad framework is equivalent to the Hilbert-space formalism wherever the two have been compared. The novel contribution is the geometric foundation and computational efficiency inherent in the tetrad framework. At present, the tetrad formalism provides no additional predictive content and the potential experimental consequences remain to be investigated more completely. It is however  worth noting that the tetrad framework does offer insight into spacetime dimension. Since tetrad generators are applied to spacetime vectors, they are necessarily $d\times d$ in $d$ dimensions. These tetrad generators are also bi-vectors in the spacetime algebra, generated by  Dirac matrices of dimension $2^{\lfloor d/2\rfloor}$. This requires $2^{\lfloor d/2\rfloor}=d$, an equation satisfied only for $d=1$, $d=2$ and $d=4$, however the case $d=1$ means $e_0$ has no orthogonal complement and so it is excluded.  Moreover, when $d=2$ there is only one operator such that $\Omega_{1}e_0=e_1$, and so no Clifford product or $\mathfrak{su}(2)$ Lie algebra properties are available.  

Several directions remain open. The extension to curved spacetime where the tetrad becomes a local frame along a worldline would test how far the tetrad framing  of spin persists. This raises a further question, since Rovelli \cite[\S 5.6.4]{Rov2004} asks whether the relational character of spacetime and that of quantum mechanics (see \cite{Rov1996} for example) share a common origin, and the construction here is suggestive in this direction. The tetrads are defined only in relation to one another, and the same map $j^{a}\Omega_{a}$ is a definite Lorentz transformation for a single momentum and a superposition of such transformations for a general spinor. The chirality distinguishing the two pseudoscalar representations $\Omega_{i}$ and $\Omega_{i}^{*}$, and its relation to the two Weyl sectors of the Dirac field, likewise merits closer study. Finally, the hydrogenic ground state provides a setting in which candidate relativistic spin operators may in principle be distinguished experimentally \cite{BAKG2014}, and the behaviour of the intrinsic operator $\gamma_{0}\Omega_{i}(b)$ in that setting remains to be determined.

\appendix
\section{Representations of Clifford algebras}\label{apx:reps}
This appendix derives an explicit matrix representation of the tetrad generators $\Omega_i$ satisfying the conditions of Lemma~\ref{lem:metric-anti-symmtry} and Theorem~\ref{thm:Clifford-product}, and establishes the unitary equivalence with the standard Pauli matrices $\sigma_i$ used in the main text.
Let us denote
\begin{equation}
\Omega_{i} = \begin{pmatrix} a_{i} & \vb{q}^{\mathrm{T}}_{i} \\ \vb{r}_{i} & \mathrm{D}_{i} \end{pmatrix}
\end{equation}
such that $\Omega_{i}^{\mathrm{T}}\eta+\eta\Omega_{i}=0$ and $\Omega_{i}^{2}=1$, then it is found that $a_{i}=0$ and $\vb{q}_{i} = \vb{r}_{i}\in S^{2}$. Upon choosing $\vb{r}_{1}=\vb{e}_{1}=(1,0,0)^{\mathrm{T}}$ and similarly choosing $\vb{r}_{2}$ and $\vb{r}_{3}$ to be the remaining cardinal directions $\vb{e}_{2}$ and $\vb{e}_{3}$, it is found that
\begin{equation}
\Omega_{1}=\begin{pmatrix} 0 & 1 & 0 & 0 \\ 1 & 0 & 0 & 0 \\ 0 & 0 & 0 & -\iu \\ 0 & 0 & \iu & 0 \end{pmatrix}
\quad
\Omega_{2}=\begin{pmatrix} 0 & 0 & 1 & 0 \\ 0 & 0 & 0 & \iu \\ 1 & 0 & 0 & 0 \\ 0 & -\iu & 0 & 0 \end{pmatrix}
\quad
\Omega_{3}=\begin{pmatrix} 0 & 0 & 0 & 1 \\ 0 & 0 & -\iu & 0 \\ 0 & \iu & 0 & 0 \\ 1 & 0 & 0 & 0 \end{pmatrix}
\end{equation}
Given the canonical time like direction $e_{0}=\begin{pmatrix} 1 & \vb{0} \end{pmatrix}^{\top}$, it is clear that the remaining orthogonal vectors in the tetrad generated by $e_{0}$ and $\Omega_{i}$ are of the form $e_{i}=\Omega_{i}e_{0}=\begin{pmatrix} 0 & \vb{e}_{i} \end{pmatrix}^{\top}$.

The involution $\Omega_{3}^{2}=1$ mean the eigenvalues of $\Omega_3$ are $\pm1$ with multiplicity two. The associated eigenvectors are
\begin{equation}
k_{1}=\frac{1}{\sqrt{2}}(e_{0}+e_{3})\quad k_2=\frac{1}{\sqrt{2}}(e_{1}-\iu e_{2})\quad k_3=\frac{1}{\sqrt{2}}(e_{1}+\iu e_{2})\quad k_{4}=\frac{1}{\sqrt{2}}(e_{0}-e_{3}).
\end{equation}
which are orthonormal under the Hermitian inner-product $k^{\dag}_{a}k_{b}=\delta_{ab}$ cf. \cite{PR1984}. Moreover, it is noted that these eigenvectors are null with respect to the Minkowski bi-linear form with $\langle k_{a},k_{a}\rangle=0$.  The unitary matrix
\begin{equation}
U=\frac{1}{\sqrt{2}}\begin{pmatrix} 1 & 0 & 0 & 1 \\ 0 & 1 & 1 & 0 \\ 0 & -\iu & \iu & 0 \\ 1 & 0 & 0 & -1 \end{pmatrix}
\end{equation}
induces the similarity transformation
\begin{equation}
U^{\dag}\Omega_{i} U= \Sigma_{i}= \diag(\sigma_{i},\sigma_{i})
\end{equation}
where $\sigma_{i}$ are the standard Pauli matrices.

Letting $\left\{\gamma^{\rm D}_{\mu}\right\}_{\mu=0}^{3}$ denote the Dirac representation then the representation $\gamma_{\mu}=U\gamma^{\rm D}_{\mu}U^{\dag}$ is of the form
\begin{equation}\label{eq:gg}
\resizebox{0.995\textwidth}{!}{
$\displaystyle
\gamma_{0} = \begin{pmatrix} 0 & 0 & 0 & 1\\0 & 0 & \iu & 0\\0 & - \iu & 0 & 0\\1 & 0 & 0 & 0 \end{pmatrix}\quad
\gamma_{1} = \begin{pmatrix} 0 & 0 & 0 & 1\\0 & 0 & \iu & 0\\0 & \iu & 0 & 0\\-1 & 0 & 0 & 0 \end{pmatrix}\quad
\gamma_{2} = \begin{pmatrix} \iu & 0 & 0 & 0\\0 & - \iu & 0 & 0\\0 & 0 & \iu & 0\\0 & 0 & 0 & - \iu \end{pmatrix}\quad
\gamma_{3} = \begin{pmatrix} 0 & -1 & 0 & 0\\1 & 0 & 0 & 0\\0 & 0 & 0 & \iu\\0 & 0 & \iu & 0 \end{pmatrix}
$%
}
\end{equation}
It is easily confirmed that
$$\gamma_{\mu}\gamma_{\nu} + \gamma_{\nu}\gamma_{\mu}=2\eta_{\mu\nu}\qquad \gamma_{i}\gamma_{j} - \gamma_{j}\gamma_{i} = -2\iu\varepsilon_{ijk}\Omega_{k}.$$

\section{Velocity potentials}\label{apx:vp}
While $\phi(x)$ is introduced to address the ambiguities introduced by non-uniform relative motion of observers and their moving frames, it is still expected that classical kinematics of uniform motion can be accommodated using the generalised velocity potential $\phi(x)$ also. In the case of uniform motion it is required that $\partial \tau(x)=u\in\mathcal{H}$ is a constant four-velocity so that
\begin{equation}
\partial_{\mu}\tau=\partial_{\mu}g(\phi) =g'(\phi)\partial_{\mu}\phi = u_{\mu} \qquad \partial_{\mu}\phi=\partial_{\mu}f(\tau)=\dot{f}(\tau)\partial_{\mu}\tau=\dot{f}(\tau)u_{\mu}
\end{equation}
Since $u_{\mu}u^{\mu}=1$, it now follows that 
\[g'(\phi)^{2}\partial_{\mu}\phi\partial^{\mu}\phi = 1 \]
so that $g'(\phi)\neq 0$ and $\partial_{\mu}\phi\partial^{\mu}\phi\neq0$ which in turn ensures $\dot{f}(\tau)\neq0$. As such, we find that
\begin{equation}
\dot{f}(\tau)g'(\phi)=\pm 1 \implies f(\cdot)=\pm g^{-1}(\cdot) 
\end{equation}
by the inverse function theorem. This ensures that solutions $\tau(x)$ of the governing equations \eqref{eq:ge1}--\eqref{eq:ge2} may be reformulated in terms of a generalised potential $\phi$, such that we may invert  between $\tau$ and $\phi$ using $f(\cdot)$ and $g(\cdot)$ which are mutually inverse. Physically, this means that each level-set $\tau(x)=\tau_{0}$ corresponds to a level-set $\phi(x)=\phi_{0}$ in such a way that $\partial \tau = g'(\phi)\partial \phi$.

The governing equations \eqref{eq:ge1}--\eqref{eq:ge2} are invariant under $\tau\to-\tau$, since $\partial_{\mu}\tau\partial^{\mu}\tau$ is quadratic and $\partial_{\mu}\partial^{\mu}\tau$ is linear and homogeneous in $\tau$. Since the level sets $\phi(x)=\phi_{0}$ and $\phi(x)^{l}=\phi_{0}^{l}$ describe the same set of spacetime points, they must correspond to the same proper time up to this sign freedom already present in the governing equations, rather than to a pair of independently-chosen inverses $g,g_{l}$. That is, we require a single function $g(\cdot)$ satisfying
\begin{equation}
\tau= g(\phi) = \epsilon\, g(\phi^{l}), \qquad \epsilon = \pm 1.
\end{equation}
Differentiating with respect to $x^{\mu}$ and using $\partial_{\mu}\phi^{l}=l\phi^{l-1}\partial_{\mu}\phi$, this gives
\begin{equation}
g'(\phi) = \epsilon\, l\, \phi^{l-1} g'(\phi^{l}),
\end{equation}
and since both $\phi$ and $\phi^{l}$ yield the same governing equations from \eqref{eq:ge1}, squaring this relation yields
\begin{equation}
g'(\phi)^{2}=l^{2}\phi^{2(l-1)}g'(\phi^{l})^{2},
\end{equation}
which is satisfied when $g'(\phi)=\phi^{r}/K$ for some constants $r$ and $K$. This now yields
\begin{equation}
\phi^{2r} = l^{2}\phi^{2(l-1)}\phi^{2rl} \implies
\begin{cases}
l=-1 \quad r=-1 \\
l=1 \quad r \text{ undetermined}.
\end{cases}
\end{equation}
The case $r=l=-1$ determines the explicit form of $g(\phi)$ since
\[g'(\phi)=\frac{1}{K\phi}\implies g(\phi)=\frac{1}{K}\ln(\phi)=\tau,\]
and substituting into $g(\phi)=\epsilon\,g(\phi^{-1})$ fixes $\epsilon=-1$, consistent with the $\tau\to-\tau$ symmetry noted above: $\phi^{-1}$ is simply the sign-reversed representation of the same congruence.

\section*{Funding}
This research did not receive any specific grant from funding agencies in the
public, commercial, or not-for-profit sectors.

\section*{Data availability statement}
No new data were created or analysed in this study.

\section*{Conflict of interest}
The author declares no conflict of interest.

\section*{Acknowledgments}
The author is grateful to the reviewers for several helpful comments and suggestions. 

\section*{References}
\providecommand{\newblock}{}

\end{document}